\documentclass[runningheads]{llncs}
\usepackage[margin=1in]{geometry}

\usepackage[numbers]{natbib}

\usepackage[T1]{fontenc}
\usepackage{amsmath,amsfonts,amssymb}
\usepackage{mathtools}
\usepackage{hyperref}
\usepackage{nameref}
\usepackage{cleveref}
\usepackage[noend,ruled]{algorithm2e}
\usepackage{dsfont}
\usepackage{longtable}
\usepackage{graphicx}
\usepackage{xcolor}
\usepackage{caption}
\usepackage{tikz}
\usepackage{todonotes}
\usepackage{wrapfig}
\usepackage{caption} 
\crefname{algocf}{alg.}{algs.}
\Crefname{algocf}{Algorithm}{Algorithms}
\crefname{fact}{fact}{facts}
\Crefname{fact}{Fact}{Facts}
\spnewtheorem{fact}{Fact}{\itshape}{\normalfont}

\makeatletter
\let\c@lemma\c@theorem
\let\c@proposition\c@theorem
\let\c@corollary\c@theorem
\let\c@definition\c@theorem
\let\c@remark\c@theorem
\let\c@example\c@theorem
\let\c@fact\c@theorem
\makeatother

\let\origproof\proof
\let\origendproof\endproof
\renewenvironment{proof}{\origproof}{\leavevmode\unskip\nobreak\hfill\squareforqed\origendproof}

\newcommand{\Oh}{\mathcal{O}}
\DeclarePairedDelimiter\ceil{\lceil}{\rceil}
\DeclarePairedDelimiter\floor{\lfloor}{\rfloor}
\newcommand{\ignore}[1]{}
\newcommand{\planar}{{\textsc{Planar}}}

\DeclareMathOperator{\E}{\mathbb{E}}
\newcommand{\mname}{ARRV}
\newcommand{\adj}[3]{#1^{#3}_{#2}}

\newcommand{\lefty}[2]{\adj{#1}{#2}{-}}
\newcommand{\W}{\widehat{W}}

\newcommand{\qsarg}{\phi}
\newcommand{{\qsl}}{\texttt{QuantSel}}

\usepackage{tikz}
\usetikzlibrary{backgrounds}
\usetikzlibrary{arrows}
\usetikzlibrary{shapes,shapes.geometric,shapes.misc}

\tikzstyle{tikzfig}=[baseline=-0.25em,scale=0.5]

\pgfkeys{/tikz/tikzit fill/.initial=0}
\pgfkeys{/tikz/tikzit draw/.initial=0}
\pgfkeys{/tikz/tikzit shape/.initial=0}
\pgfkeys{/tikz/tikzit category/.initial=0}

\pgfdeclarelayer{edgelayer}
\pgfdeclarelayer{nodelayer}
\pgfsetlayers{background,edgelayer,nodelayer,main}

\tikzstyle{none}=[inner sep=0mm]

\newcommand{\tikzfig}[1]{%
{\tikzstyle{every picture}=[tikzfig]
\IfFileExists{#1.tikz}
  {\input{#1.tikz}}
  {\tikz[baseline=-0.5em]{\node[draw=red,font=\color{red},fill=red!10!white] {\textit{#1}};}}%
}%
}

\newcommand{\ctikzfig}[1]{%
\begin{center}\rm
  \input{#1.tikz}
\end{center}}

\tikzstyle{every loop}=[]

\tikzstyle{circle}=[fill=black, draw=black, shape=circle]
\tikzstyle{empty circle}=[fill=white, draw=black, shape=circle]
\tikzstyle{diamond}=[fill={rgb,255: red,213; green,213; blue,213}, draw=black, shape=rectangle, aspect=4, scale=0.4]
\tikzstyle{rect}=[fill=white, draw=black, shape=rectangle]
\tikzstyle{selected}=[fill={rgb,255: red,81; green,181; blue,0}, draw=black, shape=rectangle]

\tikzstyle{normal}=[-, fill=none]
\tikzstyle{dotted}=[-, dashed]
\tikzstyle{pointer}=[->]
\tikzstyle{measure}=[<->, draw={rgb,255: red,0; green,46; blue,255}]
\tikzstyle{probspace}=[draw={rgb,255: red,2; green,82; blue,255}, <->]
\tikzstyle{Wcandidate}=[->, draw={rgb,255: red,28; green,113; blue,0}]
\tikzstyle{nonWcandidate}=[draw={rgb,255: red,211; green,0; blue,0}, ->]
\tikzstyle{filled}=[-, fill={rgb,255: red,162; green,162; blue,162}, draw=none, tikzit draw={rgb,255: red,6; green,6; blue,6}]
\tikzstyle{measuredotted}=[draw={rgb,255: red,0; green,46; blue,255}, dashed, <->]
\tikzstyle{measuresingle}=[draw={rgb,255: red,0; green,46; blue,255}, ->]
\tikzstyle{dotted measure}=[-, draw={rgb,255: red,0; green,53; blue,245}, dashed]
\tikzstyle{pointerdotted}=[->, dashed, draw=red]
\tikzstyle{outerbox}=[-, draw={rgb,255: red,0; green,46; blue,255}]
\tikzstyle{innerbox}=[-, draw={rgb,255: red,200; green,80; blue,0}]

\begin{document}

\title{Multiwinner Voting with Spatial Preferences under Incomplete Information}

\author{Drew Springham\inst{1} \and Edith Elkind\inst{2} \and Bart de Keijzer\inst{1} \and Maria Polukarov\inst{1}}
\authorrunning{D. Springham et al.}
\institute{King's College London, London, UK \and Northwestern University, Evanston, IL, USA}

\maketitle

\begin{abstract}
In multiwinner elections with many candidates, as in participatory budgeting or large-scale recommendation, voters cannot plausibly evaluate every candidate, yet standard proportional-fairness guarantees such as EJR+ are stated for fully specified approval ballots.
We ask whether strong proportional representation can still be guaranteed while eliciting only a little from each voter.
We study this in a spatial model, the \emph{Axis-aligned Random Rectangle Voter} (\mname) model, in which candidates occupy a $d$-dimensional issue space and each voter approves an axis-aligned hyper-rectangle: a tolerance interval on every issue.
Preferences are revealed only through \planar\ queries, each comparing a voter's tolerance to a candidate on a single issue.
We give an algorithm returning an EJR+ committee for \emph{any} distribution over rectangular preferences, using only $\Oh(d\log dk)$ \planar\ queries per voter in expectation given a sufficiently large electorate, independent of the number of candidates $m$, where $d$ is the number of issues and $k$ the committee size.
The algorithm rests on a dimension-agnostic \emph{verify-or-fallback} framework whose query cost is governed by two properties supplied by interchangeable modules.
We describe such modules, yielding end-to-end guarantees for known, unknown, and smooth distributions.
\end{abstract}

%%%%%%%%%%%%%%%%%%%%%%%%%%%%%%%%%%%%%%%%%%%%%%%%%%%%%
\section{Introduction}
Multiwinner elections arise in settings where the pool of candidates is so large that voters cannot reasonably assess every option, such as participatory budgeting, recommender systems, and large-scale committee selection. Real-world participatory-budgeting elections have featured more than 150 candidate projects \citep{Faliszewski2023ParticipatoryAnalysis}, while Pol.is (an online platform where users vote on participant-submitted comments to surface representative viewpoints \citep{small2021polis}) has hosted debates, such as a citizens' assembly on climate legislation in Austria, generating more than 1{,}000 comments \citep{TheComputationalDemocracyProject2022TheAustria}. At this scale it is unrealistic to expect voters to evaluate every candidate.
Yet the standard axiomatic guarantees in multiwinner voting, such as the proportional fairness notions PJR+ and EJR+~\cite{DBLP:conf/sigecom/Brill023}, are stated for fully specified approval ballots.
This raises a question: \emph{can we guarantee strong proportional fairness while eliciting a small amount of information from each voter?}

We study this question in a spatial model.
Candidates and voters live in an \emph{issue space} $[0,1]^d$: each of the $d$ dimensions can be thought of as an issue, each candidate has a fixed position on every issue, and a voter approves a candidate exactly when it is acceptable to her on every issue. 
Concretely, each voter's acceptable region is an axis-aligned hyper-rectangle: on each issue she has a lower and an upper tolerance, and she approves the candidates that fall within her range on \emph{every} issue.
As an illustration, in participatory budgeting each project sits in a space of policy issues (say, environmental--economic and centralised--community-led), and a resident supports the projects that are not too extreme for her on any issue, i.e.\ those lying within her acceptable band on every axis.
We call the resulting random model the \emph{Axis-aligned Random Rectangle Voter} (\mname) model. 
To elicit preferences we use \planar\ queries: a \planar\ query fixes one issue and asks whether the voter's entire acceptable range lies to one side of a given threshold on that issue, e.g.\ \emph{``would you reject every project more economically focused than candidate $c$?''}

\subsubsection{Our results and techniques.}
Our main contribution is an affirmative answer to the question above.
In the \mname\ model, we guarantee an EJR+ committee for \emph{any} distribution over rectangular preferences while eliciting, on average, only $\Oh(d\log dk)$ \planar\ queries per voter, in expectation \emph{independent of the number of candidates $m$}, where $d$ is the number of issues and $k$ the committee size.
The key component is a dimension-agnostic \emph{verify-or-fallback} framework, which we pair with interchangeable modules to obtain end-to-end guarantees across known, unknown, and smooth distributions.
Concretely:
\begin{itemize}
    \item \textbf{A verify-or-fallback framework (\Cref{sec:ddimframework}).}
    Given a guessed committee and per-issue estimates of the preference distribution, the framework elicits every voter on a small, fixed set of candidates and checks whether the guess is certifiably EJR+.
    If the check passes it returns that committee at only $\Oh(d\log dk)$ \planar\ queries per voter; otherwise it elicits the voters in full and runs a full-information EJR+ rule.
    Correctness is therefore \emph{unconditional} (the outcome is always EJR+) while the cheap cost is achieved exactly when two properties hold: the guessed committee is in fact likely to satisfy EJR+ (property~(PW)), and the per-issue distribution estimates are accurate (property~(PF)).
    \item \textbf{Modules supplying these two properties (\Cref{sec:wselection,sec:festimation}).}
    We give interchangeable modules that establish what the framework relies on: W-selection modules that produce a committee likely to satisfy EJR+ ((PW)), and $\widehat{F}$-estimation modules that estimate the per-issue distributions ((PF)), for both known and unknown distributions.
    Combined with the framework (\Cref{sec:synthesis}, \Cref{tab:results}), they yield EJR+ committees at an amortized $\Oh(d\log dk)$ queries per voter once the electorate is large enough; when the distribution is unknown, a trade-off appears between the size of electorate and how tightly we cap each voter's queries: capping \emph{every} voter's load to $\Oh(d\log dk)$ with high probability can require an electorate as large as $m\log m$, linear in the candidate count.
    \item \textbf{A smooth special case (\Cref{sec:lipschitz}).}
    Under a reasonable smoothness (Lipschitz) assumption on the unknown distribution, we do better: each sampled voter is elicited on a candidate set sized by the smoothness rather than by $m$, so the electorate depends on $m$ only as $\log\log m$, against the $m\log m$ the general unknown case needs for a comparable guarantee.
\end{itemize}

\subsubsection{Related work.}
Approval-based multiwinner voting has been studied extensively (see \citep{DBLP:series/sbis/LacknerS23} for an overview), and a range of proportionality axioms have been proposed: JR \citep{DBLP:journals/scw/AzizBCEFW17}, PJR \citep{DBLP:journals/ai/SanchezFernandezELGFBS26}, EJR \citep{DBLP:journals/scw/AzizBCEFW17}, FJR \citep{DBLP:conf/sigecom/PetersS20}, Core \citep{DBLP:journals/scw/AzizBCEFW17}, and the strengthenings PJR+ and EJR+ \citep{DBLP:conf/sigecom/Brill023}, most under the Hare quota, though the Droop quota has also been considered \citep{DBLP:conf/wine/CaseyE25}.
We target EJR+ \citep{DBLP:conf/sigecom/Brill023}, guaranteed by rules such as GJCR \citep{DBLP:conf/sigecom/Brill023}, which we adapt for committee selection and use as the fallback in our framework, and the Method of Equal Shares \citep{DBLP:conf/sigecom/PetersS20}.
Whether a Core outcome always exists remains open \citep{DBLP:journals/scw/AzizBCEFW17,DBLP:series/sbis/LacknerS23}.

We assume voters have structured, spatial preferences: each approves an axis-aligned rectangle in $d$-dimensional issue space \citep{enelow1984spatial,repec:elg:eechap:15584_10}.
In one dimension this recovers Candidate-Interval preferences \citep{DBLP:conf/ijcai/ElkindL15}; while such structure can make some rules tractable, many remain NP-hard already at $d=2$ \citep{DBLP:conf/aaai/GodziszewskiB0F21}.

We connect preference elicitation to proportional representation.
Existing elicitation work in multiwinner voting instead optimises other objectives (minimax regret \citep{DBLP:conf/ijcai/LuB13}, diversity \citep{DBLP:journals/corr/abs-2506-10843}, or social welfare \citep{DBLP:conf/sigecom/MandalSW20}), and elicitation has been studied for single-winner \citep{DBLP:conf/aaai/ConitzerS02a,DBLP:conf/sigecom/MeirLR14,DBLP:conf/atal/DeyB15,DBLP:conf/uai/ZhaoLWKMSX18} and ordinal rules, where information-theoretic limits constrain which rules can be computed from few per-voter queries \citep{DBLP:conf/sigecom/0002HT24}.
A complementary line takes partial approval ballots as given, without querying, asking which candidates are \emph{possibly} or \emph{necessarily} in a proportional committee \citep{DBLP:journals/ai/ImberIBK25}, including when voters and candidates lie in $d$-dimensional Euclidean space \citep{DBLP:conf/aaai/ImberIBSK24}; we instead actively \emph{query} voters and take a distributional rather than worst-case view.

The closest work to ours is \cite{DBLP:journals/corr/abs-2510-11625}, which studies the one-dimensional case (CI preferences under the \emph{Random Interval Voter} model) and achieves proportional representation at $\mathcal{O}(\log k)$ interval queries per voter in expectation.
Our work improves on it in three respects: it targets mainly the structured RIV model rather than a general joint distribution over endpoints; it attains only PJR+ and $2$-EJR+ on the RIV family and $2$-PJR+ for general one-dimensional distributions (where $2$-EJR+ and $2$-PJR+ weaken the group-size threshold from $n\ell/k$ to $2n\ell/k$) rather than exact EJR+; and it assumes the distribution is known exactly.
Our general W-selection modules address all three limitations, and \Cref{sec:onedim} gives a particularly clean one-dimensional rule that generalises its quantile construction to exact EJR+.
Also targeting justified representation through voter queries, \citep{Halpern2026} consider unrestricted approval preferences and give adaptive algorithms achieving $(1+\varepsilon)$-approximate EJR together with lower bounds on the number of voters sampled.
Exploiting spatial structure and a richer query model, we instead return an \emph{exact} EJR+ committee whose per-voter query complexity is, in expectation and with high probability, independent of $m$, whereas there the query count grows linearly in $m$.
Finally, the bounded-density (Lipschitz) assumption we use in \Cref{sec:lipschitz} parallels fair division, where smoothness of the density likewise controls computational efficiency: under a monotone-likelihood-ratio (hence Lipschitz) condition, envy-free divisions can be computed to arbitrary precision and Nash social welfare admits an FPTAS \citep{DBLP:journals/mor/BarmanR22}, with the smoothness constant entering only logarithmically, as in our net and grid construction.
\subsubsection{Outline}
In \Cref{sec:prelims}, we introduce notation and definitions that will be used throughout. In \Cref{sec:ddimframework}, we formally introduce the ARRV random voter model, the \qsl\ algorithm, the framework, the conditions (PW), (PF), and prove the framework's correctness and query complexity under these conditions. In \Cref{sec:wselection}, we discuss methods for creating a committee that satisfies (PW), both when the underlying distribution is known and unknown. In \Cref{sec:festimation}, we discuss how to estimate the marginal cumulative distribution functions in each dimension,  satisfying (PF). We then discuss how to achieve (PW),(PF) in the special case when the underlying distribution is Lipschitz in \Cref{sec:lipschitz}. We then demonstrate in \Cref{sec:synthesis} how these methods can be combined. We conclude in \Cref{sec:conclusion}.
%%%%%%%%%%%%%%%%%%%%%%%%%%%%%%%%%%%%%%%%%%%%%%%%%%%%%
\section{Preliminaries}
\label{sec:prelims}

All our indexing of ordered sets (or lists) starts at 0.
We write $X[i:j]\vcentcolon=\{X[t]:i\leq t <j\}$.
We use the notation $[t]\vcentcolon=\{i\in \mathbb{N}:1\leq i \leq t\}$ and $\mathds{1}[\Phi]$ as the indicator function that event $\Phi$ has occurred.
We say an event occurs \emph{with high probability (w.h.p.)} when its failure probability is $o(1)$ as $m\to\infty$; every such guarantee in this paper achieves $O(1/\log m)$.
\begin{definition}[Approval-based MWV election]
An \emph{MWV election} is a tuple $E=(V,C,k,A)$, where $V$ is a set of $n$ voters, $C$ is a set of $m\geq 2$ candidates, $k\in \mathbb{N}$ is a total number of candidates to elect, and $A:V\rightarrow 2^C$ is a function that maps each voter to the set of candidates she approves.
For each $H\subseteq V$, we write $A(H)=\bigcap_{v\in H}A(v)$.
The primary task associated with an MWV election is to select a set, or committee, of {\em winners} $W\subseteq C$ with $|W|=k$.
\end{definition}

\begin{definition}[EJR+]
A committee $W$ satisfies \emph{EJR+} \cite{DBLP:conf/sigecom/Brill023} if $|W|\leq k$ and for every $\ell \in [k]$ and every group $H\subseteq V$ that satisfies $|H|\geq n\ell /k$ and $A(H)\setminus W\neq\varnothing$, some voter $v\in H$ has $|W\cap A(v)|\geq \ell$.
\end{definition}
The weaker notion PJR+~\cite{DBLP:conf/sigecom/Brill023}, which for every such group requires only $|W\cap \bigcup_{v\in H}A(v)|\geq \ell$, is implied by EJR+; in this work we study EJR+.
We assess a committee through its \emph{cohesive groups} $N_{c,\ell}\coloneqq\{v : c\in A(v)\wedge|A(v)\cap W|<\ell\}$ for $c\in C\setminus W$ and $\ell\in[k]$: so $W$ provides EJR+ if and only if $|N_{c,\ell}|<n\ell/k$.
Our fallback elicits every voter in full and then runs the Greedy Justified Candidate Rule (GJCR, \Cref{alg:GJCR} in \Cref{app:algorithms}, \citep{DBLP:conf/sigecom/Brill023}): iterating $\ell$ from $k$ down to $1$, it adds any $c\in C\setminus W$ with $|N_{c,\ell}|\geq n\ell/k$ to $W$; we defer its role in the framework to \Cref{sec:ddimframework}.
\begin{fact}[\cite{DBLP:conf/sigecom/Brill023}]
\label{fact:GJCR}
Given an MWV election specified by full approval ballots, GJCR returns a committee of size $k$ satisfying EJR+.
\end{fact}
Our query-complexity bounds rely on two standard tail inequalities.
The first controls the deviation of a sum of bounded independent variables; the second inequality controls the \emph{uniform} error of an empirical CDF.
\begin{fact}[Hoeffding's inequality \citep{hoeffding1963probability}]
\label{fact:hoeffding}
Let $X_1,\dots,X_\nu$ be independent random variables with $X_i\in[0,1]$.
Then for every $t>0$,
$
    \Pr\left[\left|\sum_{i=1}^\nu X_i-\E\sum_{i=1}^\nu X_i\right|\geq \nu t\right]\leq 2e^{-2\nu t^2}
    .
$
\end{fact}

\begin{fact}[Dvoretzky--Kiefer--Wolfowitz (DKW) \citep{Massart1990}]
\label{fact:dkw}
Let $\widehat{F}_\nu$ be the empirical CDF of $\nu$ i.i.d. samples drawn from a distribution with CDF $F$.
Then for every $\varepsilon>0$,
$
    \Pr\left[\sup_{x\in\mathbb{R}}\,|\widehat{F}_\nu(x)-F(x)|\geq\varepsilon\right]\leq 2e^{-2\nu\varepsilon^2}.
$
\end{fact}

%%%%%%%%%%%%%%%%%%%%%%%%%%%%%%%%%%%%%%%%%%%%%%%%%%%%%
\section{The Model and the Verification Framework}
\label{sec:ddimframework}
We now present our spatial model and the verify-or-fallback framework underlying our results.
Both are stated for general dimension $d$; the one-dimensional case (\Cref{sec:onedim}) is the special case $d=1$.
We consider $d$-dimensional spatial preferences, where voters face multiple issues and have preferences over each.
We focus on voters with {\em Axis-Aligned Hyper-Rectangular (AAHR) preferences}:
\begin{definition}
We say that a multiwinner voting election is $d$-dimensional if $C\subset[0,1]^d$ and if for each voter $v\in V$, there exists some region $A^*(v)\subseteq [0,1]^d$ such that $A(v)=C\cap A^*(v)$.
In $d$-dimensional space, for any $\underline{a}\leq\underline{b}\in[0,1]^d$, $\prod_{i\in[d]}[a_i,b_i]=\{\underline{x}\in[0,1]^d:\forall i\in[d],a_i\leq x_i\leq b_i\}$ is called an Axis Aligned Hyper-Rectangle.
For a $d$-dimensional election, we say that voters have {\em Axis Aligned Hyper-Rectangular (AAHR) preferences} if for all $v\in V$, $A^*(v)$ is an axis aligned hyper-rectangle.
\end{definition}
The interpretation for AAHR preferences is that we have $d$ issues, and on each issue, a voter has a minimum and maximum range on that issue that they find tolerable.
Candidates have a fixed position on each of these issues, and thus a voter approves of the candidates who are tolerable to that voter on each issue.
We model voter preferences using the {\em Axis-aligned Random Rectangle Voter model} (\mname):
\begin{definition}[\mname]
\label{def:generald}
A finite set $C\subset [0,1]^d$, a probability distribution $\mathcal{D}$ over $d$-dimensional axis-aligned hyper-rectangles,
and $n\in\mathbb{N}$ define an Axis-aligned Random Rectangle Voter model
(\mname) as follows: we draw $n$ samples $A^*(1), \dots, A^*(n)\sim \mathcal{D}$,
let $V=[n]$, and define the approval set of voter $v\in V$
as $A(v)=A^*(v)\cap C$.
\end{definition}
For a committee $W$, we write $p_{c,\ell}\coloneqq\Pr_{A^*(v)\sim\mathcal{D}}[v\in N_{c,\ell}]$ for the corresponding single-voter probability; as voters are i.i.d., $\E|N_{c,\ell}|=n\,p_{c,\ell}$. Note that $p_{c,\ell}$ depends on $W$ and is non-increasing as $W$ grows (adding candidates shrinks $N_{c,\ell}$); in module proofs, $p_{c,\ell}$ refers to the final returned committee unless a superscript $p^W_{c,\ell}$ specifies otherwise.
Each voter's approval set is determined by per-dimension endpoints $a^i_v\leq b^i_v$ with $A^*(v)=\prod_{i\in[d]}[a^i_v,b^i_v]$.
Write
$
    F^i_a(x)\coloneqq\Pr[a^i_v\leq x], F^i_b(x)\coloneqq\Pr[b^i_v\leq x]
$
for the per-dimension CDFs of the lower and upper endpoints.

We elicit preferences through \planar\ queries, which tests a voter's approval rectangle against an axis-aligned threshold on one issue.
\begin{definition}[\planar\ query]
\label{def:planar}
Given a candidate $c\in C$, a voter $v\in V$, an axis $i\in[d]$, and a direction $\lhd\in\{< ,> \}$, the query \planar$(c, v,i,\lhd)$ returns $\mathds{1}[A^*(v)\subseteq \{\underline{x}\in [0,1]^d:x_i\lhd c_i\}]$.
\end{definition}
\planar$(c, v,i,\lhd)$ effectively asks voter $v$ ``Is your entire acceptable range on issue $i$ located to the left/right ($\lhd$) of $c$'s position on issue $i$?''
A query may probe only the coordinate of an \emph{actual candidate}, never an arbitrary threshold; this keeps the elicitation model strictly weaker and more realistic than one permitting free thresholds.
We now introduce the sub-routine that determines which candidate coordinates to query; it is reused in both the verification framework and the W-selection modules.
\subsubsection{The {\qsl} function}
\label{sec:quantsel}
Throughout the paper we rely on the function {\qsl}, which places a small set of candidates at evenly-spaced quantile levels of a CDF.
Given a candidate set $C$, a spacing $\qsarg>0$, and a CDF $G$, ${\qsl}(C,\qsarg,G)$ (\Cref{alg:quantile-select}) scans the quantile levels $\qsarg,2\qsarg,\ldots$ from left to right and, at each level, selects the still-available candidate with the smallest $G$-value at or above that level, breaking ties by an arbitrary fixed rule (\Cref{rem:strictmono}).
We shall reuse {\qsl} in several places later; we state here a bound that the verification framework (\Cref{sec:ddimverify}) and Lipschitz estimator (\Cref{sec:lipshitz}) rely on.

\begin{lemma}[Coverage]
\label{lem:lefty}
Let $S={\qsl}(C,\qsarg,G)$ and $c\in C\setminus S$.
Either there exists $u\in S$ with $G(u)\leq G(c)<G(u)+\qsarg$ or $G(c)<\qsarg$.
\end{lemma}
In words, $S$ leaves no gap wider than $\qsarg$ in $G$-value: every unselected candidate either lies in the bottom band below the first level $\qsarg$, or is covered from below by a selected point within $\qsarg$ of its own $G$-value.
\begin{proof}
Let $r^*=\lfloor G(c)/\qsarg\rfloor$.
If $r^*=0$ then $G(c)<\qsarg$ and we are done.
Otherwise, at step $r^*$, candidate $c$ is still available ($c\notin S$) and satisfies
$G(c)\geq r^*\qsarg$, so $c\in Q_{r^*}$.
The candidate $u$ selected at step $r^*$ minimises $G$ over $Q_{r^*}$, so
$G(u)\leq G(c)$.
Since $G(c)<(r^*+1)\qsarg$, we get $G(c)-G(u)\leq G(c)-r^*\qsarg<\qsarg$.
\end{proof}

\subsection{Verification framework}
\label{sec:ddimverify}
We now state the verify-or-fallback framework in full generality.
Before the formal statement, we trace the procedure end-to-end.
The framework (\Cref{alg:multidim}) receives a guessed committee $\W$, per-dimension CDF estimates $\widehat{F}^i_a$, $\widehat{F}^i_b$ (the true CDFs when the distribution is known, or the output of an estimation module otherwise), and a quantile spacing $\Delta$ (fixed in \Cref{thm:multidim}), and proceeds in four steps.
\begin{enumerate}
    \item \textbf{Build the query set $P$.} For each dimension $i$, two {\qsl} calls place candidates at the evenly-spaced quantile levels of $1-\widehat{F}^i_a$ and of $\widehat{F}^i_b$; $P$ is the union of these per-dimension points together with $\W$, so $P\subseteq C$.
The two calls per dimension are complementary, controlling the spacing of $P$ relative to the left- and right-endpoint CDFs respectively (\Cref{lem:quantile-gaps}).
Each dimension contributes at most $2/\Delta$ points, so $|P|\leq k+2d/\Delta=\Oh(d^2k^2)$ under \Cref{thm:multidim}'s spacing $\Delta=\Omega(1/(dk^2))$.
    \item \textbf{Query every voter on $P$.} \texttt{resolve} (\Cref{alg:querying}) runs one binary search per dimension using \planar\ queries, returning an outer bracket $\hat{A}(v)\supseteq A(v)$.
This outer bracket agrees with the true approval set on $P$: $A(v)\cap P=\hat{A}(v)\cap P$ (\Cref{prop:approximateapprovals}, illustrated in \Cref{fig:bracketing}).
    \item \textbf{Count possible witnesses.} For each $c\notin\W$ and $\ell\in[k]$, let $s_{\ell,c}$ be the number of voters with $c\in\hat{A}(v)$ and $|\W\cap\hat{A}(v)|<\ell$.
These are the voters that could witness an EJR+ violation against $\W$ at $(c,\ell)$.
    \item \textbf{Verify or fall back.} If $|\W|\leq k$ and every $s_{\ell,c}<n\ell/k$, output $\W$; otherwise discard $\W$, elicit all voters fully, and return an EJR+ committee via GJCR.
\end{enumerate}
\begin{figure}[t]
\centering
\begin{minipage}[t]{0.49\textwidth}
\begin{algorithm}[H]
\SetAlgoSkip{}
\SetKwProg{Fn}{Function}{:}{}
\SetKwFunction{Fqsl}{\qsl}
\Fn{\Fqsl{$C,\phi,G$}}{
    $S\gets\varnothing$\;
    \For{$r=1,2,\ldots,\lceil 1/\qsarg\rceil-1$}{
        $Q_r\gets\{c\in C\setminus S:G(c)\geq r\qsarg\}$\;
        \If{$Q_r\ne\varnothing$}{
            $u\gets\arg\min_{c\in Q_r}G(c)$, break ties by an arbitrary fixed rule\;
            $S\gets S\cup\{u\}$\;
        }
    }
    \Return{$S$}\;
}
\caption{Quantile selection.}
\label{alg:quantile-select}
\end{algorithm}
\begin{algorithm}[H]
\SetNoFillComment
\SetKwFunction{FMulti}{framework}
\SetKwFunction{FVerify}{verify}
\SetKwFunction{FQuery}{resolve}
\SetKwProg{Fn}{Function}{:}{}
\Fn{\FMulti{$\W,\widehat{F},\Delta,d,k,V,C$}}{
    \lIf{\FVerify{$\W,\widehat{F},\Delta,C$}}{\KwRet{$\W$}}
    \lFor{$v\in V$}{$A(v)\gets$\FQuery{$v,C$}}
    \KwRet{GJCR on full info (\Cref{alg:GJCR})}
}
\Fn{\FVerify{$\W,\widehat{F},\Delta,C$}}{
    $P\gets \W$\;
    \For{$i\in[d]$}{
        $P\gets P\cup{\qsl}(C,\Delta,1-\widehat{F}^i_a)\cup{\qsl}(C,\Delta,\widehat{F}^i_b)$\;
    }
    \lFor{$v\in V$}{$\hat{A}(v)\gets$\FQuery{$v,P$}}
    \For{$c\notin\W,\ell\in[k]$}{$s_{\ell,c}\gets|\{v:c\in\hat{A}(v),|\W\cap\hat{A}(v)|<\ell\}|$}
    \KwRet{Whether $|\W|\leq k$ and all $s_{\ell,c}< n\ell/k$}
}
\caption{The verification framework.}
\label{alg:multidim}
\end{algorithm}
\end{minipage}
\hfill
\begin{minipage}[t]{0.49\textwidth}
\begin{algorithm}[H]
\SetNoFillComment
\SetKwFunction{FAxisQuery}{outerbound}
\SetKwFunction{FQuery}{resolve}
\SetKwProg{Fn}{Function}{:}{}
\Fn{\FAxisQuery{$v,P,i,\lhd$}}{
    \If{$\lhd=<$}{
        Sort $P$ by coord. $i$ in increasing order\;
    }\Else{
       Sort $P$ by coord. $i$ in decreasing order\;
    }
    Add un-queried sentinels $P[-1],P[|P|]$ at $\mp\infty$, where \planar\ is \texttt{false}/\texttt{true} by $A^*(v)\subseteq[0,1]^d$\;
    $x\gets 0$\;
    $y\gets |P|$\;
    \While{$x<y$}{
        $\mu\gets \floor{(x+y)/2}$\;
        \If{\planar$(P[\mu],v,i,\lhd)$}{$y\gets\mu$}\Else{$x\gets\mu+1$}
    }
    \KwRet{$P[x-1],P[x]$}
}
\Fn{\FQuery{$v,P$}}{
    \For{$i\in [d]$}{
        $\check{a}^i_v,\hat{a}^i_v\gets$ \FAxisQuery{$v,P,i,>$}\;
        $\check{b}^i_v,\hat{b}^i_v\gets$ \FAxisQuery{$v,P,i,<$}\;
    }
    $\hat{A}(v)=C\cap\prod_{i\in[d]}(\hat{a}^i_v,\hat{b}^i_v)$\;
    $\check{A}(v)=C\cap\prod_{i\in[d]}[\check{a}^i_v,\check{b}^i_v]$\tcc*{$\check{A}$ only used in proofs.}
    \KwRet{$\hat{A}(v)$}
}
\caption{$d$-dimensional querying.}
\label{alg:querying}
\end{algorithm}
\end{minipage}
\end{figure}
The guarantee our framework provides has two parts. 
\emph{Correctness is unconditional:} since $\W\subseteq P$ and $\hat{A}$ is exact on $P$, every genuine member of a cohesive group $N_{c,\ell}$ is counted in $s_{\ell,c}$, so passing the test certifies EJR+ no matter how poor $\W$ or the estimates are, and the fallback committee is EJR+ by construction. 
\emph{Query-efficiency is conditional:} We later introduce conditions to keep the false-positive rate of being added to $s_{l,c}$ low, so the test passes, and the expensive fallback is avoided with high probability.
This is precisely what the modules of \Cref{sec:wselection,sec:festimation} are built to supply.

\subsubsection{Correctness of framework}
\begin{figure}[t]
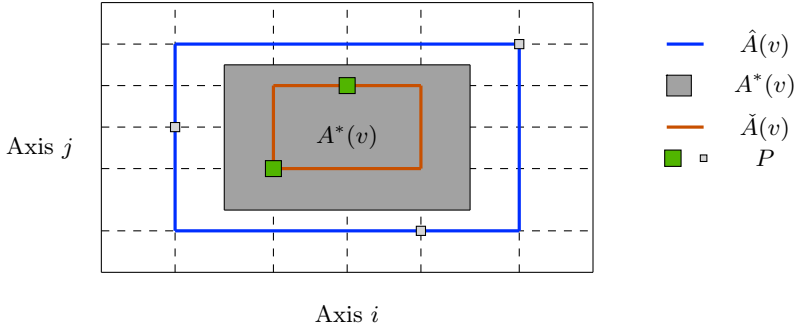

\centering
\begin{minipage}[t]{0.6\textwidth}
\centering
\ctikzfig{bucketing_ddim_paper}
\end{minipage}\hfill
\begin{minipage}[t]{0.35\textwidth}
\captionsetup{type=figure}
\caption{The bracketing produced by \Cref{alg:querying} in
$d=2$ dimensions. The voter's true approval rectangle $A^*(v)$ (grey) is
between inner bracket $\check{A}(v)$ and outer bracket $\hat{A}(v)$. Candidates
in $P$ (gridlines mark projections onto each axis) are classified
by querying: those known to be approved by $v$ (green) lie inside
$\check{A}(v)$, while those known to be rejected (diamonds) lie on the boundary
of $\hat{A}(v)$.}
\label{fig:bracketing}
\end{minipage}
\end{figure}
We prove correctness now. First, \texttt{resolve} produces valid brackets, and that the outer bracket $\hat{A}(v)$ is tight on $P$.
\newcommand{\approxapprovalsstatement}{$\check{A}(v)\subseteq A(v)\subseteq \hat{A}(v)$, and $A(v)\cap P=\hat{A}(v)\cap P$.}
\begin{proposition}
\label{prop:approximateapprovals}
\approxapprovalsstatement
\end{proposition}
For the full proof of \Cref{prop:approximateapprovals}, see \Cref{sec:ddimquery}.
See \Cref{fig:bracketing} for a diagram.
This holds regardless of how $P$ is actually constructed.

\begin{proposition}[Query cost of \texttt{resolve}]
\label{prop:queryingqueries}
For any voter $v$ and query set $P$, a call to \texttt{resolve}$(v,P)$ makes at most $2d\ceil{\log_2(|P|+1)}=\Oh(d\log|P|)$ \planar\ queries.
\end{proposition}
\begin{proof}
Each call to \texttt{outerbound} is a binary search on $P$ sorted along one axis: $|P|$ halves every iteration, so the loop runs at most $\ceil{\log_2(|P|+1)}$ times and issues exactly one \planar\ query per iteration; note the two $\pm\infty$ sentinels are never queried.
\texttt{resolve} invokes \texttt{outerbound} twice per dimension (once in each direction) i.e. $2d$ times, for a total of at most $2d\ceil{\log_2(|P|+1)}$ queries.
\end{proof}
\begin{theorem}[Correctness]
\label{thm:correctness}
\Cref{alg:multidim} provides EJR+.
In the worst case it issues $\Oh(d\log m)$ queries per voter, for any spacing $\Delta$.
\end{theorem}
\begin{proof}
If GJCR is used, the outcome is EJR+ by \Cref{fact:GJCR}, so suppose $\W$ is returned, which means $|\W|\leq k$ and $s_{\ell,c}<n\ell/k$ for all $\ell,c$.
For any voter $v$ with $c\in A(v)$ and $|\W\cap A(v)|<\ell$, so $v\in N_{c,\ell}$, \Cref{prop:approximateapprovals} together with $\W\subseteq P$ gives $c\in\hat{A}(v)$ and $\W\cap\hat{A}(v)=\W\cap A(v)$.
Therefore $v$ is counted in $s_{\ell,c}$ and $|N_{c,\ell}|\leq s_{\ell,c}<n\ell/k$ for every $c,\ell$, so $\W$ provides EJR+.
The verification pass uses candidate set $P\subseteq C$, costing $\Oh(d\log|P|)=\Oh(d\log m)$ per voter; if \texttt{verify} returns false, the fallback re-elicits on $C$, a further $\Oh(d\log m)$.
Either way at most $\Oh(d\log m)$ queries per voter, regardless of $\Delta$.
\end{proof}
\subsubsection{Query-efficiency of framework}
The framework is query-efficient when its two inputs $\W,\widehat{F}$ are ``good''; we now make this precise.
The first property asks that $\W$ is a good guess at an EJR+ committee (the probability that a voter falls in a given cohesive group is below $\ell/k$ by a margin $\delta>0$), and the second that the per-dimension CDF estimates are accurate on $C$.
\begin{definition}[PW]
\label{def:PW-1d}
A committee $\W$ satisfies \emph{(PW) with margin $\delta$} if for every $c\in C\setminus\W$ and $\ell\in[k]$,
$
    p_{c,\ell}\leq \frac{\ell}{k} - \delta.
$
\end{definition}
\begin{definition}[PF]
\label{def:PF}
Per-dimension estimates $\widehat{F}^i_a$, $\widehat{F}^i_b$ satisfy \emph{(PF) with error $\varepsilon$} if, for every $i\in[d]$, and every $c\in C$,
$
    |\widehat{F}^i_a(c)-F^i_a(c)|\leq\varepsilon,$ $ |\widehat{F}^i_b(c)-F^i_b(c)|\leq\varepsilon.
$
When the distribution is known, (PF) holds with $\varepsilon=0$.
\end{definition}
We first show that when (PF) holds, we can bound the true CDF gaps at the $P$-brackets of any $c\notin P$.

\newcommand{\quantilegapsstatement}{Let $P_i\coloneqq{\qsl}(C,\Delta,1-\widehat{F}^i_a)\cup{\qsl}(C,\Delta,\widehat{F}^i_b)$ be the candidates added to $P$ on dimension $i$ inside \texttt{verify}, where $\widehat{F}^i_a$, $\widehat{F}^i_b$ satisfy (PF) with error $\varepsilon$.
For any candidate $c\notin P$ and dimension $i\in[d]$, let $c^+_i\coloneqq\min\{p_i:p\in P_i,\,p_i\geq c_i\}$ and $c^-_i\coloneqq\max\{p_i:p\in P_i,\,p_i\leq c_i\}$ be the closest $P_i$-coordinates on either side of $c_i$, with the convention $c^+_i=+\infty$ and $F^i_a(+\infty)=F^i_b(+\infty)=1$ if no $P_i$-candidate has $p_i\geq c_i$, and symmetrically $c^-_i=-\infty$ with $F^i_a(-\infty)=F^i_b(-\infty)=0$ if none has $p_i\leq c_i$.
Then
$
    F^i_a(c^+_i)-F^i_a(c_i)<\Delta+2\varepsilon$, and $ F^i_b(c_i)-F^i_b(c^-_i)<\Delta+2\varepsilon.
$}
\begin{lemma}
\label{lem:quantile-gaps}
\quantilegapsstatement
\end{lemma}
Intuitively, the bracket $(c_i,c^+_i)$ lies within a single {\qsl} gap of $\widehat{F}^i_a$, whose estimated width is below $\Delta$ by coverage (\Cref{lem:lefty}); passing from $\widehat{F}^i_a$ to the true $F^i_a$ via (PF) costs a further $2\varepsilon$, and the left bracket is symmetric. The full proof is deferred to \Cref{app:multidim}.

The second ingredient for the main theorem is a union-bound compression.
The outer bracket $\hat{A}(v)$ returned by \texttt{outerbound} has its endpoints in $P^{(i)}\cup\{\pm\infty\}$, so whether $c_i$ falls inside the bracket depends only on which \emph{cell} of the partition of $\mathbb{R}$ induced by $P^{(i)}$ contains $c_i$.
Two candidates lying in the same cell on every axis therefore share the same witness count $s_{\ell,c}$, reducing the number of distinct counts from $m$ to the number of occupied cells $\Lambda$.
\begin{lemma}[$P$-cells]
\label{lem:cells}
For each axis $i$, let $P^{(i)}=\{p_i:p\in P\}$ be the set of $i$-th coordinates of the query candidates ($|P^{(i)}|\leq|P|$).
Call two candidates $c,c'\in C\setminus\W$ \emph{$P$-equivalent} if, on every axis $i$, $c_i$ and $c'_i$ lie in the same interval of the partition of $\mathbb{R}$ induced by $P^{(i)}$ (the $|P^{(i)}|$ coordinate points together with the open intervals between consecutive ones).
Then:
\begin{enumerate}
    \item[(i)] If $c$ and $c'$ are $P$-equivalent then $c\in\hat{A}(v)\iff c'\in\hat{A}(v)$ for every voter, hence $s_{\ell,c}=s_{\ell,c'}$ for all $\ell$.
    \item[(ii)] The number of non empty $P$-equivalence classes is at most $\Lambda\coloneqq\min\left\{m,(2|P|+1)^d\right\}$.
\end{enumerate}
\end{lemma}
\begin{proof}
For part~(i): by \Cref{alg:querying} the outer bracket on axis $i$ is the open interval $(\hat{a}^i_v,\hat{b}^i_v)$, whose endpoints are returned by the binary search \texttt{outerbound} over $P$ sorted on axis $i$ and so lie in $P^{(i)}\cup\{\pm\infty\}$.
Hence whether $c_i\in(\hat{a}^i_v,\hat{b}^i_v)$ depends only on which cell of $P^{(i)}$ contains $c_i$, so it is constant across $P$-equivalent candidates on every axis.
Taking the product over $i\in[d]$ gives $c\in\hat{A}(v)\iff c'\in\hat{A}(v)$; and since $s_{\ell,c}=|\{v:c\in\hat{A}(v),\,|\W\cap\hat{A}(v)|<\ell\}|$ depends on $c$ only through the set $\{v:c\in\hat{A}(v)\}$, we get $s_{\ell,c}=s_{\ell,c'}$.

For part~(ii): on each axis the $|P^{(i)}|\leq|P|$ coordinates partition $\mathbb{R}$ into at most $|P|+1$ open intervals together with the $|P|$ coordinate points, i.e.\ at most $2|P|+1$ cells; combining $d$ axes gives at most $(2|P|+1)^d$ classes, of which at most $m$ are occupied by candidates.
\end{proof}
Given (PW) margin $\delta$ and (PF) error $\varepsilon$, the framework runs \texttt{verify} with the quantile spacing
$
    \Delta=\frac{\delta}{4d}-\varepsilon.
$
\newcommand{\multidimstatement}{Suppose $\W$ satisfies (PW) with margin $\delta$ and the per-dimension CDF estimates $\widehat{F}^i_a$, $\widehat{F}^i_b$ satisfy (PF) with error $\varepsilon$. Let spacing $\Delta=\delta/(4d)-\varepsilon$. When $\Delta\geq \alpha/2dk^2$ for some constant $\alpha>0$, and $n=\Omega(k^4\log(\Lambda\log m))$ where $\Lambda$ is the $P$-cell count of \Cref{lem:cells}, \Cref{alg:multidim}  uses $\Oh(d\log dk)$ queries in expectation per voter.}
\begin{theorem}[Query efficiency]
\label{thm:multidim}
\multidimstatement
\end{theorem}
We sketch the proof; the full proof can be found in \Cref{app:multidim}
\begin{proof}[sketch]
Every voter is asked the $\Oh(d\log|P|)=\Oh(d\log dk)$ queries of one \texttt{resolve} pass over $P$, and only a verification failure adds the $\Oh(d\log m)$ of a full re-elicitation; so the expected per-voter cost is $\Oh(d\log dk)+\Pr(\text{fail})\cdot\Oh(d\log m)$, and it remains to bound $\Pr(\text{fail})$.

Fix $c\in C\setminus\W$ and $\ell\in[k]$. A voter is counted in $s_{\ell,c}$ either by genuinely lying in the cohesive group $N_{c,\ell}$ with probability $p_{c,\ell}\leq\ell/k-\delta$ by (PW), or as a bracketing false positive ($c\in\hat A(v)\setminus A(v)$), which forces an endpoint into one of the bracketing gaps $(c_i,c^+_i)$, $(c^-_i,c_i)$ on some axis; \Cref{lem:quantile-gaps} bounds each gap's mass by $\Delta+2\varepsilon$, so a union over the $d$ axes caps the false-positive probability by $2d(\Delta+2\varepsilon)$. 
Hence a voter is counted with probability at most $\ell/k-\gamma$, where $\gamma=\delta-2d(\Delta+2\varepsilon)\geq\alpha/k^2$ by the hypothesis. \Cref{fact:hoeffding} then gives $\Pr(s_{\ell,c}\geq n\ell/k)\leq\exp(-2n\alpha^2/k^4)$. By \Cref{lem:cells}, $s_{\ell,c}=s_{\ell,c'}$ for all $P$-equivalent $c'$, so a union bound over the $\Oh(\Lambda k)$ cell--level pairs yields $\Pr(\text{fail})\leq\Lambda k\exp(-2n\alpha^2/k^4)$. At $n=\Omega(k^4\log(\Lambda\log m))$ this is $\Oh(1/\log m)$, so the fallback cost is absorbed into the $\Oh(d\log dk)$ verification cost and the per-voter cost is $\Oh(d\log dk)$. Note that $\log\Lambda=\Oh(d\log (dk))$ when $\Delta\geq \alpha/2dk^2$.
\end{proof}
Having shown that the framework is query-efficient under the conditions (PW) and (PF), we now show how to achieve these properties: we provide W-selection modules supplying (PW) (\Cref{sec:wselection}) and $\widehat{F}$-estimation modules supplying (PF) (\Cref{sec:festimation}).
Each module is reported in a common form: to meet its guarantee it may draw a \emph{pool} of voters, summarised by its \emph{pool size} (how many it draws) and \emph{per-voter query load} (the most queries any pooled voter is asked), written $n_W,q_W$ for a W-selection module and $n_F,q_F$ for an $\widehat{F}$-estimation module.
A module reading off a known distribution draws no pool ($n_W=q_W=0$); a sampling one incurs both.
Only pooled voters are elicited (every other voter pays just the verification pass), and the two budgets are composed in \Cref{sec:synthesis}.

\begin{remark}[Two senses of the per-voter bound]
\label{rem:twosenses}
\Cref{thm:multidim} bounds the per-voter cost \emph{in expectation} (amortized over the $n$ voters).
When every pooled voter's load is $q_W,q_F=\Oh(d\log dk)$, the stronger guarantee holds: \emph{every} voter is asked $\Oh(d\log dk)$ queries w.h.p., bar the rare fallback. A heavier load ($\Oh(d\log m)$, or $\Oh(d^2\log(Kdk))$ under Lipschitz) leaves only the amortized bound.
The modules of \Cref{sec:wselection,sec:festimation,sec:lipschitz} realise both senses; the stronger every-voter guarantee requires a larger minimum electorate. We tag each result accordingly.
\end{remark}
%%%%%%%%%%%%%%%%%%%%%%%%%%%%%%%%%%%%%%%%%%%%%%%%%%%%%
\section{W-selection modules}
\label{sec:wselection}
Recall that the framework consumes a committee $\W$ that is already likely to be EJR+, in the sense of (PW).
We now describe two interchangeable W-selection modules: PGJCR (\Cref{sec:ejrhp}), an adapted version of GJCR (\Cref{alg:GJCR}, \citep{DBLP:conf/sigecom/Brill023}), for when the distribution $\mathcal{D}$ is known, and the query-based NGJCR (\Cref{sec:ngjcr}), an adapted version of Noisy GJCR \citep{DBLP:conf/sigecom/Brill023}, for when it is not.
We analyse each only up to its (PW) guarantee and its query cost; the end-to-end results follow by plugging them into the framework in \Cref{sec:synthesis}.
All proofs in this section are deferred to \Cref{sec:pgjcrproof}.
\subsection{Distribution-input PGJCR}
\label{sec:ejrhp}

Our first W-selection module is PGJCR (\Cref{alg:PGJCR}), which constructs a committee satisfying (PW) given the distribution $\mathcal{D}$, without making any queries.
\Cref{alg:PGJCR} is based on GJCR \cite{DBLP:conf/sigecom/Brill023}: given a distribution over the $2^m$ ballots, we greedily take candidates $c$ if, in expectation, the number of voters who approve of $c$ and fewer than $\ell$ candidates in $W$ is greater than $n\ell/(k+1)$, with $\ell$ decreasing from $k$ to $1$.
This algorithm is not specific to a spatial setting, and works whenever we have any probability distribution over approval ballots.

\newcommand{\pgjcrstatement}{When $\mathcal{D}$ is known and the $A(v)$ are i.i.d.\ from $\mathcal{D}$, PGJCR consumes no pool and makes no queries ($n_W=0$, $q_W=0$), and supplies (PW) with margin $\delta=1/(k(k+1))$ deterministically (failure probability $p_W=0$), returning a committee of at most $k$ candidates.}
\begin{theorem}
\label{thm:pgjcr}
\pgjcrstatement
\end{theorem}
For the size bound, we distribute a total budget of $k+1$ proportionally across the ballots and charge each elected candidate a cost of $1$ split among the ballots that contributed to its selection; because PGJCR guarantees a sufficiently large contributing mass at each step, and no voter can exhaust her budget, the total budget (decreasing by exactly $1$ per iteration from an initial $k+1$) admits at most $k$ selections.
For the (PW) guarantee: by construction, PGJCR only adds $c$ to $W$ when $p_{c,\ell}>\ell/(k+1)$; hence for any $c\notin W$ and $\ell\in[k]$, $p_{c,\ell}\leq\ell/(k+1)=\ell/k-\ell/(k(k+1))\leq\ell/k-1/(k(k+1))$, satisfying (PW) with margin $\delta=1/(k(k+1))$.
\subsection{Query-based W-selection (NGJCR)}
\label{sec:ngjcr}

When the distribution $\mathcal{D}$ is unknown, we cannot run PGJCR directly.
NGJCR adapts the Noisy GJCR of \citep{DBLP:conf/sigecom/Brill023} by replacing each exact-expectation check in PGJCR with a sample-based estimate.
For each level $\ell$, it partitions $C\setminus W$ into batches of at most $\beta$ candidates, draws $h_1$ fresh voters per batch (each queried via \texttt{resolve}), and adds $c$ to $W$ when the empirical witness fraction $\zeta_c/h_1$ exceeds
$
    q^* = \frac{(2k+1)\ell}{2k(k+1)} - \frac{\delta_1}{2},
$
the midpoint between the ``tiny'' upper bound $\ell/(k+1)$ and the ``large'' lower bound $\ell/k-\delta_1$.
The formal pseudocode is deferred to \Cref{app:algorithms}.

\newcommand{\ngjcrstatement}{Let $p_W>0$, let $\beta\geq k$ be the size of the candidate partitions, and let $0<\delta_1\leq 1/(k(k+1))-\alpha/k^2$ for some constant $\alpha>0$.
When $\mathcal{D}$ is unknown, NGJCR (\Cref{alg:noisygjcr}) consumes a pool of $n_W=\Theta(h_1km/\beta)$ voters where $h_1=k^4/\alpha^2 \log(4mk/p_W)$, each asks $q_W=\Oh(d\log\beta)$ \planar\ queries, and supplies (PW) with margin $\delta_1$ with probability $1-p_W$.}
\begin{theorem}
\label{thm:ngjcr}
\ngjcrstatement
\end{theorem}
In particular, with $\beta=\Theta(k^5)$ each pooled voter answers $q_W=\Oh(d\log k)$ queries, pool $n_W=\Theta(m\log(m/p_W))$, while with $\beta=\Theta(m)$ each answers $q_W=\Oh(d\log m)$, pool $n_W=\Theta(k^5\log(m/p_W))$).
For the (PW) guarantee: call $c$ \emph{tiny} at level $\ell$ if $p_{c,\ell}\leq\ell/(k+1)$ and \emph{large} if $p_{c,\ell}\geq\ell/k-\delta_1$.
The hypothesis $\delta_1\leq 1/(k(k+1))-\alpha/k^2$ ensures $q^*$ sits $\Omega(\alpha/k^2)$ above the tiny bound and $\Omega(\alpha/k^2)$ below the large bound.
\Cref{fact:hoeffding} bounds the probability that any tiny candidate clears $q^*$ or any large candidate fails to, across all $mk$ (candidate, level) pairs, by $p_W$.
On the success event no tiny candidate is selected and every large candidate is selectable; any $c\notin W$ went unselected hence is not large, so $p_{c,\ell}\leq\ell/k-\delta_1$, confirming (PW) with margin $\delta_1$.
The size bound follows from the budget argument of \Cref{thm:pgjcr}, since each selected candidate has true mass $>\ell/(k+1)$ at its selection level.

\begin{remark}
\label{rem:ngjcr-framework}
Unlike PGJCR, which supplies (PW) deterministically, \Cref{thm:ngjcr} supplies (PW) only with probability $1-p_W$.
The framework (\Cref{thm:multidim}) contributes a second failure probability $\eta$: even when $\W$ does satisfy (PW), \texttt{verify} may reject it and trigger the expensive fallback.
The fallback therefore occurs with probability at most $p_W+\eta$ by a union bound, and setting both to $\Oh(1/\log m)$ keeps the total $\Oh(1/\log m)$ without changing the asymptotic per-voter query count.
When the $\widehat{F}$-estimation module is also randomised (\Cref{sec:festimation}), a third failure probability $p_F$ enters by the same argument.
\end{remark}
\subsection{A Special Case: the One-Dimensional Setting}
\label{sec:onedim}
The one-dimensional case $d=1$ admits a particularly clean W-selection rule that closely mirrors the quantile construction of \cite{DBLP:journals/corr/abs-2510-11625}: both place committee members near equally spaced quantile positions on the line.
Where \cite{DBLP:journals/corr/abs-2510-11625} selects the candidates closest to the marked points $i/(k+1)$ under the combined mixture CDF, we instead make a single {\qsl} call on the left-endpoint marginal CDF $\widehat{F}_a$, i.e.\ $\W={\qsl}(C,1/(k+1),\widehat{F}_a)$.
Despite the resemblance, quantising the left-endpoint marginal rather than the RIV mixture CDF is what carries essentially the same construction to exact EJR+.

Because the selection depends on the voters only through $\widehat{F}_a$, the accuracy of that single estimate is all that governs how evenly the chosen candidates are spread, a phenomenon special to one dimension, where a (PF) guarantee alone already yields a (PW) guarantee.
\Cref{thm:1dpry} makes this precise: the (PF) error $\varepsilon$ of $\widehat{F}_a$ passes directly into the margin of (PW).

\newcommand{\onedprystatement}{Given a CDF estimate $\widehat{F}_a$ satisfying (PF) with error $\varepsilon$ (\Cref{def:PF}), the rule $\W={\qsl}(C,1/(k+1),\widehat{F}_a)$ consumes no pool and makes no queries ($n_W=0$, $q_W=0$), and supplies (PW) with margin $\delta=\frac{1}{k(k+1)}-2\varepsilon$ deterministically ($p_W=0$), provided $\varepsilon<\frac{1}{2k(k+1)}$.
}
\begin{theorem}[1-D W-selection]
\label{thm:1dpry}
\onedprystatement
\end{theorem}

%%%%%%%%%%%%%%%%%%%%%%%%%%%%%%%%%%%%%%%%%%%%%%%%%%%%%
\section{$\widehat{F}$-Estimation Modules}
\label{sec:festimation}
Recall that the framework also needs accurate per-dimension CDF estimates: property (PF).
When $\mathcal{D}$ is known these are computed exactly; otherwise we estimate them with a single ``batch estimator'' that queries each pooled voter on a batch of at most $\beta$ candidates.
The batch size $\beta\in[1,m]$ trades pool size against per-voter load: $\beta=m$ is a single batch covering all of $C$ at load $\Oh(d\log m)$, while smaller batches cut the load to $\Oh(d\log\beta)$ at the cost of a larger pool.
Each candidate is assigned to one of the $m/\beta$ batches,\footnote{The partition into batches is arbitrary: any assignment works.} and batch $B$ is served by a fresh pool $V_B$ of $h_2$ voters resolved only on $B$, so each answers $\Oh(d\log\beta)$ queries and the per-batch estimate $\widehat{F}^i_{a,B}(c)=\frac{1}{h_2}|\{v\in V_B:a^i_v\leq c_i\}|$ (for $c\in B$, axis $i$) is a genuine empirical CDF of $h_2$ i.i.d.\ samples.
We report, on each axis, the running maximum $\widehat{F}^i_a(c):=\max\{\widehat{F}^i_{a,B(c')}(c'):c'\in C,\ c'_i\leq c_i\}$, the least non-decreasing function dominating the per-batch values (a valid CDF estimate), and same for $\widehat{F}^i_b$.
\newcommand{\logkblockstatement}{The estimate $\widehat{F}^i_a$ satisfies $\Pr[\max _{c\in C}|\widehat{F}^i_a(c)-F^i_a(c)|>\varepsilon]\leq \frac{2m}{\beta}\exp(-2h_2\varepsilon^2)$, and for $\widehat{F}^i_b$.}
\begin{lemma}
\label{lem:logk-block}
    \logkblockstatement
\end{lemma}
Each per-batch estimate is an empirical CDF, so \Cref{fact:dkw} plus a union bound over the $m/\beta$ batches gives the stated probability. The full proof is deferred to \Cref{app:logk}.
\begin{theorem}
\label{thm:logk-estimator}
    Let $p_F>0$ and $1\leq\beta\leq m$.
When $\mathcal{D}$ is unknown, this estimator consumes a pool of $n_F=(m/\beta)\,h_2$ voters, with $h_2=\left(\log\left(\frac{4dm}{\beta p_F}\right)\right)/(2\varepsilon^2)$ voters per batch, each asked $q_F=\Oh(d\log\beta)$ \planar\ queries, and supplies (PF) with error $\varepsilon$ with probability $1-p_F$.
\end{theorem}
\begin{proof}
Union bounding \Cref{lem:logk-block} over the $d$ dimensions and both endpoints, the probability that any estimate is more than $\varepsilon$ from its true CDF is at most $\frac{4dm}{\beta}\exp(-2h_2\varepsilon^2)$, which is at most $p_F$ once $h_2\geq\frac{\log(4dm/(\beta p_F))}{2\varepsilon^2}$; this gives (PF) with error $\varepsilon$ with probability $1-p_F$.
For the query cost, each pooled voter is resolved on a single batch of $\beta$ candidates, so by \Cref{prop:queryingqueries} answers $\Oh(d\log\beta)$ \planar\ queries.
\end{proof}

\section{Smooth electorates: Lipschitz distributions}
\label{sec:lipschitz}
The unknown-distribution modules of \Cref{sec:ngjcr,sec:festimation} keep the required electorate dependent on $m$: an adversarial distribution can hide a cohesive group behind an arbitrarily fine distinction between candidate positions, forcing elicitation everywhere.
We propose that real electorates are seldom so adversarial: nearby candidates should be approved by almost the same voters, so approval statistics vary smoothly across candidate space.
We formalise this with a bounded density on the rectangle-endpoint space, which has two consequences: the marginal endpoint CDFs are $K$-Lipschitz (\Cref{lem:lipschitz-cdf}), and so is the cohesive-group mass as a function of candidate position (\Cref{lem:lip-stability}).
Each lets a module elicit voters on only an $m$-independent set of candidates: a \emph{net} for W-selection (\Cref{sec:netngjcr}) and a quantile \emph{grid} for $\widehat{F}$-estimation (\Cref{sec:lipshitz}), which combine into an EJR+ guarantee.
All proofs in this section are deferred to \Cref{app:lipschitz}.

\begin{definition}
\label{def:lipschitz}
    We say that an \mname\ is \emph{$K$-Lipschitz} if its distribution $\mathcal{D}$ over axis-aligned hyper-rectangles is absolutely continuous, with a density $f$ on the endpoint space $\mathcal{R}=\{(\underline a,\underline b)\in[0,1]^{2d}:\underline a\leq\underline b\}$ identifying a rectangle $R=\prod_{i\in[d]}[a_i,b_i]$ with its endpoint vector $(a_1,b_1,\dots,a_d,b_d)$ such that $f(R)\leq K$ for every $R\in\mathcal{R}$.
\end{definition}
This density bound is exactly what is needed for the name: bounding $f$ by $K$ directly makes the endpoint CDFs $K$-Lipschitz.

\newcommand{\lipcdfstatement}{If an \mname\ is $K$-Lipschitz, then for any $i\in[d]$,
    $|F^i_a(x)-F^i_a(y)|\leq K|x-y|$ and similar for $b$.}
\begin{lemma}
\label{lem:lipschitz-cdf}
    \lipcdfstatement
\end{lemma}

\subsection{W-selection via a candidate net}
\label{sec:netngjcr}
We first construct a committee satisfying (PW). The enabling fact is that the cohesive-group mass is Lipschitz in the candidate, uniformly in the committee built so far and in the level, so NGJCR can be run eliciting voters on only an $m$-independent \emph{net} of candidates.

\newcommand{\lipstabilitystatement}{Let an \mname\ be $K$-Lipschitz.
For every $W\subseteq C$, every $\ell\in[k]$, and all $c,c'\in[0,1]^d$,
$
    \left|p_{c,\ell}-p_{c',\ell}\right|\ \leq\ K\,\|c-c'\|_1 .
$}
\begin{lemma}[Lipschitz stability of cohesive mass]
\label{lem:lip-stability}
\lipstabilitystatement
\end{lemma}
Intuitively, write $p_{c,\ell}=\Pr[\{c\in A(v)\}\cap J]$ with the committee event $J=\{|A(v)\cap W|<\ell\}$, which involves no candidate position. For two positions $c,c'$ the event $J$ is common to both, so $p_{c,\ell}$ and $p_{c',\ell}$ differ only through the approval events $\{c\in A(v)\}$ and $\{c'\in A(v)\}$; since $c\in A(v)$ means $a^i_v\leq c_i\leq b^i_v$ in every coordinate $i$, \Cref{lem:lipschitz-cdf} bounds this difference coordinate by coordinate, summing to $K\|c-c'\|_1$.
We discretise the candidates at a resolution fine enough that stability error is negligible against the (PW) margin.

\begin{definition}[Candidate net]
\label{def:net}
For $\rho>0$, partition $[0,1]^d$ into the cells $\prod_{i\in[d]}[\,j_i\rho,(j_i{+}1)\rho)$, $j\in\mathbb{Z}_{\geq0}^d$.
A \emph{$\rho$-net} is a set $\mathcal{N}\subseteq C$ containing exactly one candidate from each cell that meets $C$, together with the map $\mathrm{rep}:C\to \mathcal{N}$ sending each candidate to the representative of its cell.
By construction $\|c-\mathrm{rep}(c)\|_\infty<\rho$, hence $\|c-\mathrm{rep}(c)\|_1<d\rho$, and
$
    |\mathcal{N}|\leq\min\left(m,\ \lceil 1/\rho\rceil^{\,d}\right).
$
\end{definition}
\begin{corollary}
\label{cor:net-transfer}
Combining \Cref{lem:lip-stability} and \Cref{def:net}, for every $\ell\in[k]$ and $c\in C$, $\left|p_{c,\ell}-p_{\mathrm{rep}(c),\ell}\right|\leq K d\rho$.
\end{corollary}

\Cref{alg:netngjcr} (\Cref{app:algorithms}) runs NGJCR on the net $\mathcal{N}$, scoring each candidate $c$ by its representative's count $\zeta_{\mathrm{rep}(c)}$ and splitting $\mathcal{N}$ into batches of size $\beta$ served by fresh pools, so a voter answers $\Oh(d\log\beta)$ queries.

\newcommand{\netngjcrstatement}{Let $p_W>0$, let $\alpha>0$ be a constant with $0<\delta_1\leq 1/(k(k+1))-\alpha/k^2$, and let $k\leq\beta\leq|\mathcal{N}|$.
For an \mname\ that is $K$-Lipschitz, set $\rho=\alpha/(4Kdk^2)$, so $|\mathcal{N}|\leq\min\left(m,\lceil 4Kdk^2/\alpha\rceil^{d}\right)$, and $h_1=\lceil(8k^4/\alpha^2)\log(4k|\mathcal{N}|/p_W)\rceil$.
When $\mathcal{D}$ is unknown, net-NGJCR consumes a pool of $n_W=\Theta\left(k\,h_1\lceil|\mathcal{N}|/\beta\rceil\right)$ voters, each asked $q_W=\Oh(d\log\beta)$ \planar\ queries, and supplies (PW) with margin $\delta_1$ with probability at least $1-p_W$.
For a single batch $\beta=|\mathcal{N}|$, $q_W=\Oh(d^2\log(Kdk))$ and $n_W=\Oh\left(k^5\left(d\log(Kdk)+\log(1/p_W)\right)\right)$.
}
\begin{theorem}[Net-NGJCR]
\label{thm:netngjcr}
\netngjcrstatement
\end{theorem}
Net-NGJCR scores each candidate by its representative's empirical witness fraction, incurring two errors against the true mass: the net-transfer error $\leq Kd\rho$ (\Cref{cor:net-transfer}) and the sampling error over $h_1$ voters.
Choosing $\rho$ so $Kd\rho$ equals the sampling tolerance $\tau=\alpha/(4k^2)$ keeps both below $\tau$; as $q^*$ sits a margin $\Gamma\geq 2\tau$ from each of the ``tiny'' ($\ell/(k+1)$) and ``large'' ($\ell/k-\delta_1$) bounds, the combined $2\tau$ error never selects a tiny candidate nor misses a large one, giving (PW).
The size bound $k$ is inherited from \Cref{thm:pgjcr}.

\subsection{$\widehat{F}$-estimation on a quantile grid}
\label{sec:lipshitz}
It remains to supply the framework's second input: per-dimension CDF estimates accurate to $\varepsilon$ on all of $C$, again from an $m$-independent number of queries.
Since a $K$-Lipschitz CDF varies slowly, it suffices to estimate it on a coarse uniform grid $T\coloneqq{\qsl}(C,\varepsilon/2K,\text{Uniform})$ of only $|T|=\Oh(K/\varepsilon)$ points (independent of $m$) and interpolate: we estimate the empirical CDF at the grid points (splitting $T$ into batches served by fresh pools and queried with \texttt{resolve}, exactly as the batch estimator of \Cref{thm:logk-estimator}), after which the Lipschitz bound of \Cref{lem:lipschitz-cdf} fills the gaps between them.

\newcommand{\lipestimatorstatement}{Let $p_F>0$ and $1\leq\beta\leq|T|$.
For an \mname\ that is $K$-Lipschitz, form the grid $T\coloneqq{\qsl}(C,\varepsilon/2K,\text{Uniform})$ (so $|T|\leq 2K/\varepsilon$) and set $h_2=2\log\left(\frac{8dK}{\varepsilon\beta p_F}\right)/\varepsilon^2$.
When $\mathcal{D}$ is unknown, this estimator consumes a pool of $n_F=\lceil|T|/\beta\rceil\,h_2$ voters, each asked $q_F=\Oh(d\log\beta)$ \planar\ queries, and supplies (PF) with error $\varepsilon$ with probability $1-p_F$.}
\begin{theorem}
\label{thm:lipschitz-estimator}
    \lipestimatorstatement
\end{theorem}
The error at a candidate $c$ splits in two. \emph{Interpolation}: $c$ lies within $\varepsilon/2K$ of its nearest grid point, across which the $K$-Lipschitz CDF moves by at most $\varepsilon/2$ (\Cref{lem:lipschitz-cdf}). \emph{Estimation}: each grid point reports an empirical CDF of $h_2$ i.i.d.\ samples, kept within $\varepsilon/2$ uniformly by DKW (\Cref{fact:dkw}) from an $m$-independent $h_2=\Oh(\varepsilon^{-2}\log(dK/\varepsilon\beta p_F))$. The spacing $\varepsilon/2K$ balances the two into total error $\varepsilon$; a running maximum stitches the batches into one monotone CDF.
%%%%%%%%%%%%%%%%%%%%%%%%%%%%%%%%%%%%%%%%%%%%%%%%%%%%%
\section{Putting It Together}
\label{sec:synthesis}
Each end-to-end guarantee composes one W-selection and one $\widehat{F}$-estimation module through the verification framework.
Since each module consumes a pool of voters at some per-voter query load, composing them composes two budgets: the pools, which set the minimum electorate $n$, and the per-voter load.
Outside the pools every voter pays only the verification pass's $\Oh(d\log dk)$ \planar\ queries; we report the \emph{amortized expected} complexity (total queries over $n$), in which each fixed pool's $\Oh(Pq/n)$ share is dominated by the verification pass once $n$ is large, so the average tends to $\Oh(d\log dk)$.
Whether this also holds for \emph{every} voter w.h.p.\ is the distinction made precise in \Cref{rem:twosenses}.

\begin{fact}[Amortized query complexity]
\label{fact:amortized}
Run the framework with a W-selection module that queries a pool of $n_W$ voters at most $q_W$ times each and succeeds with probability $p_W$, an $\widehat{F}$-estimation module that queries a disjoint pool of $n_F$ voters at most $q_F$ times each and succeeds with probability $p_F$, and verification on the remaining voters.
Then the expected total number of \planar\ queries is at most
\[
    n_W q_W + n_F q_F + n\cdot\Oh(d\log dk) + (\eta+p_W+p_F)\cdot n\cdot\Oh(d\log m),
\]
where $\eta$ is the probability that verification fails. $\eta+p_W+p_F$ is a union bound over the probability that either of the modules, or the framework fail.
The last term is the fallback: on failure of a component, the framework (at most) re-elicits all $n$ voters on all $m$ candidates, $\Oh(d\log m)$ queries each.
Equivalently, the amortized expected number of queries per voter is
$
    \Oh(d\log dk) + (n_W q_W + n_F q_F)/n + (\eta+p_W+p_F)\cdot\Oh(d\log m).
$
In particular, if $n=\Omega(k^4\log(\Lambda\log m))$, $n= \Omega((n_W q_W + n_F q_F)/(d\log dk))$, and $p_W,p_F,\eta=\Oh(1/
\log m)$, the amortized expected query complexity is $\Oh(d\log dk)$ per voter.
\end{fact}

We instantiate \Cref{fact:amortized} in the settings of \Cref{tab:results}, making two standing choices forced by the framework.
The estimator accuracy is $\varepsilon=\Theta(1/(dk^2))$, the loosest meeting \Cref{thm:multidim}'s hypothesis $\delta-4d\varepsilon=\Theta(1/k^2)$ (any looser breaks verification).
The module failure probabilities are $p_W,p_F=\Theta(1/\log m)$, matching the framework's $\eta=\Oh(1/\log m)$ (\Cref{rem:ngjcr-framework}); being only logarithmically small they add merely a $\log\log m$ factor and, crucially for the Lipschitz setting (\Cref{cor:lipschitz}), reintroduce no polynomial-in-$m$ pool.
The $n$-thresholds in \Cref{tab:results} are simplified by treating $d,k$ as small (the corollaries carry the precise dependencies) and assuming $\log m\geq d\log dk$, the regime where the cell bound of \Cref{lem:cells} is effective: then $\log\Lambda\leq d\log dk$ is $m$-free, so the verification floor $\Omega(k^4\log(\Lambda\log m))$ collapses to $\Omega(\mathrm{poly}(d,k)\log\log m)$, as the Known and Lipschitz rows report.

\begin{table}[t]
\caption{End-to-end settings; the amortized expected query complexity is $\Oh(d\log dk)$ per voter throughout. The \emph{per-voter bound} column marks whether this also holds for every voter (\emph{w.h.p.}) or only on average (\emph{expected}), per \Cref{rem:twosenses}. The $n$-thresholds are simplified sufficient bounds (exact forms in \Cref{app:synthesis}).}
\label{tab:results}
\centering
\small
\begin{tabular}{llll}
\hline
Setting & Modules & per-voter bound & $n$-threshold \\
\hline
Known $\mathcal{D}$ (\Cref{cor:knowndist})              & PGJCR\,/\,exact            & w.h.p.   & $\Omega(\text{poly}(k)\log\log m)$ \\
Unknown, $\beta{=}\Theta(m)$ (\Cref{cor:logm})          & NGJCR\,/\,$\log m$ est.    & expected     & $\Omega\left(\text{poly}(d,k)\log^2 m\right)$ \\
Unknown, $\beta{=}\Theta(k^5)$ (\Cref{cor:logk})        & NGJCR\,/\,$\log k$ est.    & w.h.p.   & $\Omega\left(\text{poly}(d,k)m\log m\right)$ \\
$K$-Lipschitz $\mathcal{D}$ (\Cref{cor:lipschitz})           & net-NGJCR\,/\,grid est.   & expected   & $\Omega\left(\text{poly}(d,k)\log K\left(\log\log m+\log K\right)\right)$ \\
\hline
\end{tabular}
\end{table}

The settings trade assumptions on $\mathcal{D}$ against cost (\Cref{tab:results}): a known $\mathcal{D}$ makes selection and estimation free (every voter pays only the verification pass); an unknown $\mathcal{D}$ forces sampled pools, trading electorate against per-voter load ($\beta=\Theta(m)$ asks $\Oh(d\log m)$ each in expectation, $\beta=\Theta(k^5)$ caps every voter at $\Oh(d\log k)$ w.h.p.\ but needs $n=\Omega(m\log m)$); and smoothness removes the polynomial-in-$m$ electorate, leaving $m$-dependence only through $\log\log m$.
The proofs of the four corollaries below are deferred to \Cref{app:synthesis}.
\newcommand{\knowndiststatement}{Given $\mathcal{D}$, PGJCR with the exact per-dimension CDFs yields an EJR+ committee at $\Oh(d\log dk)$ queries per voter, for any $n=\Omega(k^4\log(\Lambda\log m))$; as no voter is pooled, this holds for every voter w.h.p.}
\begin{corollary}[Known distribution]
\label{cor:knowndist}
\knowndiststatement
\end{corollary}

\newcommand{\logmstatement}{With $\mathcal{D}$ unknown and a per-voter budget of $\Oh(\log m)$, NGJCR and the batch estimator, both at $\beta=\Theta(m)$, yield an EJR+ committee at $\Oh(d\log dk)$ amortized expected queries per voter, whenever $n=\Omega\left(k^5\log^2 m+d^2k^4\log m\log\log m\right)$.}
\begin{corollary}[Unknown distribution, $\Oh(\log m)$ budget]
\label{cor:logm}
\logmstatement
\end{corollary}

\newcommand{\logkstatement}{Suppose $m=\Omega(k^5)$. With $\mathcal{D}$ unknown and a per-voter budget of only $\Oh(\log k)$, restricting elicitation to batches of size $\beta=\Theta(k^5)$ (so $\beta\leq m$), NGJCR and the batch estimator yield an EJR+ committee at $\Oh(d\log dk)$ queries per voter (and, as every pooled voter answers only $\Oh(d\log k)$, for every voter w.h.p.), whenever $n=\Omega\left((k^4+m)\log m+(d^2m/k)\log(dm)\right)$.}
\begin{corollary}[Unknown distribution, $\Oh(\log k)$ budget]
\label{cor:logk}
\logkstatement
\end{corollary}

\newcommand{\lipschitzcorstatement}{With $\mathcal{D}$ unknown but $K$-Lipschitz, net-NGJCR and the Lipschitz grid estimator, each on a single pool ($q_W=\Oh(d^2\log(Kdk))$, $q_F=\Oh(d\log(Kdk))$), yield an EJR+ committee at $\Oh(d\log dk)$ queries per voter in expectation whenever
$$
    n=\Omega\left(d^2k^5\log K\left(\log(Kdk)+\log(d\log m)\right)\right).
$$
}
\begin{corollary}[Lipschitz distribution, single pool]
\label{cor:lipschitz}
\lipschitzcorstatement
\end{corollary}

\begin{remark}[Recovering an every-voter (w.h.p.) bound under smoothness]
\label{rem:lipschitz-blocked}
Batching the net and grid to $\beta=\Theta(k^5)$ would cap every pooled voter at $\Oh(d\log k)$, restoring the every-voter (w.h.p.)\ guarantee of \Cref{rem:twosenses}, but at a steep price: the net pool then scales with $|\mathcal{N}|=\Oh((Kdk^2)^d)$, exponential in $d$.
We therefore report the single-pool, expected-sense form, whose threshold is polynomial in $d,k$ and only polylogarithmic in $K$.
\end{remark}

\section{Conclusion}
\label{sec:conclusion}
We have studied a multiwinner voting setting where candidates and voters lie in a $d$-dimensional space.
We assume many candidates ($m$ large), and that we are able to query voters using \planar\ queries.
The question we seek to answer is whether we can still achieve justified representation in this setting without asking all voters about all candidates.

We have introduced a framework that takes as input an initial guess committee and estimates for the distribution marginals.
It always returns an EJR+ committee, and, conditional on the guess and estimates being accurate, requires only $\Oh(d\log dk)$ \planar\ queries per voter, amortized, in expectation.
This bound does not depend on $m$, the total number of candidates.
We have provided ways to obtain the guess and estimates in four settings: (i) when $\mathcal{D}$ is known; (ii) when $\mathcal{D}$ is unknown, where some voters are asked $\Oh(d\log m)$ queries; (iii) when $\mathcal{D}$ is unknown, where some voters are asked only $\Oh(d\log k)$ queries; and (iv) when $\mathcal{D}$ is smooth (Lipschitz).
In particular, in the smooth case we have shown that even when the distribution is unknown we can still obtain the result with an electorate that does not scale polynomially in $m$, depending on $m$ only through $\log\log m$.

Several questions remain open.
We would like to consider more general $d$-dimensional preferences, e.g.\ spheres or general convex shapes.
We would also like to look at whether the core can be guaranteed.
A lower bound on the amount of information required from voters in expectation would also be interesting: that is, if we elicit fewer than $d\log dk$ queries per voter, can we still guarantee EJR+? We conjecture not.
This could take the form of constructing a family of elections such that, regardless of how we elicit that information, there would always exist two elections such that there is no EJR+ committee for both elections simultaneously.
Finally, it would be interesting to see how realistic axis-aligned rectangles are with respect to real preferences.

\subsubsection*{Acknowledgments.}
This work was partly supported by the UK Research and Innovation (UKRI) Engineering and Physical Sciences Research Council (EPSRC) [grant numbers EP/X038351/1, EP/X038348/1]. This work was also partly supported by UK Research and Innovation [grant number EP/S023356/1], in the UKRI Centre for Doctoral Training in Safe and Trusted Artificial Intelligence (\href{https://safeandtrustedai.org}{safeandtrustedai.org}).

\bibliographystyle{splncs04}
\bibliography{ddim_clean}

%%%%%%%%%%%%%%%%%%%%%%%%%%%%%%%%%%%%%%%%%%%%%%%%%%%%%
\newpage
\appendix
% The body and the appendix share the raw section counter (\appendix resets it
% to 0), so hyperref would build identical PDF anchors (section.1, subsection.2.1,
% ...) for appendix and body units, making appendix links jump to the body.
% Give appendix units their own anchor names to keep every destination unique.
\renewcommand{\theHsection}{appendix.\Alph{section}}
\renewcommand{\theHsubsection}{appendix.\Alph{section}.\arabic{subsection}}

\section{Table of Notation}
\label{sec:notation}
For reference, we collect here all notation used in the paper.

\renewcommand{\arraystretch}{1.25}
\begin{longtable}{@{}l@{\hspace{1.5em}}p{0.8\textwidth}@{}}
\hline
\textbf{Symbol} & \textbf{Meaning}\\
\hline
\endfirsthead
\hline
\textbf{Symbol} & \textbf{Meaning}\\
\hline
\endhead
\hline
\endfoot
\multicolumn{2}{@{}l}{\textit{Elections, committees, and fairness}}\\
$E=(V,C,k,A)$ & A multiwinner voting (MWV) election.\\
$V,\ v$ & Set of $n$ voters; a voter $v\in V$.\\
$C,\ c$ & Set of $m$ candidates $C\subset[0,1]^d$; a candidate $c\in C$, with coordinate $c_i$ on axis $i$.\\
$n,\ m,\ k,\ d$ & Number of voters, number of candidates, committee size, number of dimensions.\\
$A,\ A(v)$ & Approval function $A:V\to 2^C$; the approval set of voter $v$.\\
$H,\ A(H)$ & A group of voters $H\subseteq V$; their common approvals $A(H)=\bigcap_{v\in H}A(v)$.\\
$\ell$ & A representation level, $\ell\in[k]$.\\
$\widehat W$ & The guessed committee supplied to the framework.\\
$N_{c,\ell}$ & Cohesive group $\{v:c\in A(v),\ |A(v)\cap W|<\ell\}$.\\
\hline
\multicolumn{2}{@{}l}{\textit{The \mname\ spatial model}}\\
$A^*(v)$ & Voter $v$'s true approval region (a hyper-rectangle).\\
$\mathcal D$ & Distribution over axis-aligned hyper-rectangles.\\
$a^i_v,\ b^i_v$ & Lower/upper endpoint of voter $v$ on axis $i$ ($a^i_v\leq b^i_v$).\\
$\underline a,\ \underline b$ & Endpoint vectors; the rectangle $\prod_{i\in[d]}[a_i,b_i]$.\\
$F^i_a,\ F^i_b$ & True per-dimension CDFs of the lower/upper endpoints ($F_a,F_b$ when $d=1$).\\
$p_{c,\ell}$ & Single-voter cohesive probability $\Pr[v\in N_{c,\ell}]$.\\
\hline
\multicolumn{2}{@{}l}{\textit{Planar queries and bracketing}}\\
$\hat A(v),\ \check A(v)$ & Outer/inner approval brackets, $\check A(v)\subseteq A(v)\subseteq\hat A(v)$.\\
$\hat a^i_v,\hat b^i_v;\ \check a^i_v,\check b^i_v$ & Outer/inner bracket endpoints on axis $i$.\\
$P$ & Query set built by \texttt{verify} ($P\subseteq C$).\\
$P_i$ & Candidates added to $P$ on axis $i$ by {\qsl} (a subset of $C$).\\
$P^{(i)}$ & Set of $i$-th coordinate values of candidates in $P$, i.e.\ $\{p_i:p\in P\}$.\\
$c^-_i,\ c^+_i$ & Closest $P_i$-coordinates bracketing $c_i$ ($c^-_i\leq c_i\leq c^+_i$).\\
$s_{\ell,c}$ & Number of counted possible EJR+ witnesses against $\widehat W$ at $(c,\ell)$.\\
\hline
\multicolumn{2}{@{}l}{\textit{The verification framework}}\\
$\Delta=\delta/(4d)-\varepsilon$ & \qsl spacing used inside \texttt{verify} (derived from $\delta,\varepsilon$).\\
$\delta$ & (PW) margin.\\
$\varepsilon$ & (PF) estimation error.\\
$\gamma=\delta-2d(\Delta+2\varepsilon)$ & Net margin after the bracketing gap loss.\\
$\alpha$ & Small positive constant in the margin condition $\gamma\geq\alpha/k^2$.\\
$\eta$ & Probability that verification fails.\\
$\Lambda=\min\{m,(2|P|+1)^d\}$ & Number of $P$-equivalence classes (cells); $\log\Lambda=O(d\log dk)$ (\Cref{lem:cells}).\\
$\Xi,\ Z$ & Event that voter $v$ is counted in $s_{\ell,c}$; its sub-event with $c\notin A(v)$.\\
\hline
\multicolumn{2}{@{}l}{\textit{W-selection and $\widehat F$-estimation modules}}\\
$\beta$ & Batch size.\\
$B,\ B_t,\ B(c)$ & A candidate batch; the $t$-th batch; the batch containing $c$.\\
$V_B$ & Voter pool serving batch $B$; the single estimation pool (case $\beta=m$).\\
$\zeta_c$ & Empirical witness count for a candidate.\\
$q^*$ & Empirical selection threshold (NGJCR).\\
$\widehat F^i_a,\ \widehat F^i_b$ & Estimated per-dimension CDFs.\\
$\widehat F^i_{a,B},\ \widehat F_B$ & Per-batch CDF estimate; batch empirical CDF (axis $i$, lower endpoint).\\
$n_W,\ q_W$ & W-selection module: pool size and per-voter query load.\\
$n_F,\ q_F$ & $\widehat F$-estimation module: pool size and per-voter query load.\\
$p_W,\ p_F$ & Failure probabilities of the W-selection / estimation modules.\\
\hline
\multicolumn{2}{@{}l}{\textit{Smooth (Lipschitz) setting}}\\
$K,\ f$ & Lipschitz constant; the density on the endpoint space ($f\leq K$).\\
$\mathcal R$ & Endpoint space $\{(\underline a,\underline b)\in[0,1]^{2d}:\underline a\leq\underline b\}$.\\
$\rho,\ \mathcal N,\ \mathcal N_t,\ \mathrm{rep}$ & Net cell resolution; the candidate net; net batches; cell-representative map $\mathrm{rep}:C\to\mathcal N$.\\
$T$ & Quantile grid ($|T|\leq 2K/\varepsilon$)\\
$\lefty{c}{T}$ & Nearest grid point in $T$ with coordinate $\leq c_i$ on each axis (left-extension of $c$ onto $T$). Contrast with $c^-_i$ (the left $P$-bracket).\\
$\tau=\alpha/(4k^2)$ & Lipschitz stability slack.\\
$\mathcal E,\ \mathcal E_i$ & Good (concentration) event; per-coordinate event.\\
\hline
\multicolumn{2}{@{}l}{\textit{PGJCR budget analysis}}\\
$\mathcal P(C)$ & Power set of $C$ (the possible ballots).\\
$X,\ p_X$ & A ballot $X\in\mathcal P(C)$; its probability $\Pr[A(v)=X]$.\\
$\mathcal X_c$ & Ballots witnessing $c$ at level $\ell$: $\{X:c\in X,\ |W\cap X|<\ell\}$.\\
$p_c$ & $\Pr[A^*(v)\in\mathcal X_c]$.\\
$\theta_i(X)$ & Budget of ballot $X$ after $i$ selections; $\theta_0(X)=(k+1)p_X$.\\
$\pi_c(X)$ & Payment $p_X/p_c$ charged to ballot $X$ for candidate $c$.\\
\hline
\multicolumn{2}{@{}l}{\textit{General mathematical and probabilistic notation}}\\
$[t]$ & The set $\{1,\dots,t\}$.\\
$X[i:j]$ & Slice notation $\{X[t]:i\leq t<j\}$; indexing starts at $0$.\\
$\mathds 1[\Phi]$ & Indicator of event $\Phi$.\\
\end{longtable}
\renewcommand{\arraystretch}{1}
\section{Omitted algorithm listings}
\label{app:algorithms}
For completeness we collect here the pseudocode deferred from the body: the fallback rule GJCR
(\Cref{alg:GJCR}), the distribution-input W-selection rule PGJCR (\Cref{alg:PGJCR}), the query-based
W-selection rule NGJCR (\Cref{alg:noisygjcr}), and its Lipschitz, net-based variant (\Cref{alg:netngjcr}).

\begin{algorithm}[p]
\SetKwFunction{FGJCR}{GJCR}
\SetKwProg{Fn}{Function}{:}{}
\Fn{\FGJCR{$C,V,k$}}{
    $W\gets \varnothing$\;
    \For{$\ell= k$ to $1$}{
        \For{$c\in C\setminus W$}{
            \lIf{$|N_{c,\ell}|\geq n\ell/k$}{
                $W\gets W\cup\{c\}$
            }
        }
    }
    \KwRet{W}
}
\caption{GJCR \citep{DBLP:conf/sigecom/Brill023}}
\label{alg:GJCR}
\end{algorithm}

\begin{algorithm}[p]
\SetKwFunction{FPGJCR}{PGJCR}
\SetKwProg{Fn}{Function}{:}{}
\Fn{\FPGJCR{$C,V,\mathcal{D},k$}}{
    $W\gets \varnothing$\;
    \For{$\ell= k$ to $1$}{
        \For{$c\in C\setminus W$}{
            \lIf{$\E[|N_{c,\ell}|]> n\ell/(k+1)$}{
                $W\gets W\cup\{c\}$
            }
        }
    }
    \KwRet{W}
}
\caption{PGJCR}
\label{alg:PGJCR}
\end{algorithm}

\begin{algorithm}[p]
\SetKwFunction{FQuery}{resolve}
\caption{Noisy Greedy Justified Candidate Rule \citep{DBLP:conf/sigecom/Brill023}, adapted to our notation}
\label{alg:noisygjcr}
$p_W,\delta_1,\beta$ provided\;
$W \gets \emptyset$\;
$h_1 \gets \frac{k^4}{\alpha^2}\log(4mk/p_W)$\;
$\ell \gets k$\;
\While{$\ell \geq 1$}{
    $q^*\gets(2k+1)\ell/(2k(k+1))-\delta_1/2$\;
    Partition $C \setminus W$ into $B_1, \dots, B_{\lceil m/\beta\rceil}$\;
    \FQuery{$v,W \cup B_t$} for each $t \in [\lceil m/\beta\rceil]$ and for $h_1$ fresh voters\;
    Assign $\zeta_c$ the number of queried voters $j$ with $c\in A(j)$ and $\lvert A(j) \cap W \rvert < \ell$, for each $c \notin W$\;
    \eIf{there is $c\notin W$: $\zeta_c/h_1 \geq q^*$}{
        $W \gets W \cup \{c\}$\;
    }{
        \While{there is no $c\notin W$: $\zeta_c/h_1 \geq q^*$}{
            $\ell \gets  \ell - 1$\;
            $q^*\gets(2k+1)\ell/(2k(k+1))-\delta_1/2$\;
            }
    }
}
\Return{$W$}\;
\end{algorithm}

\begin{algorithm}[p]
\SetKwFunction{FQuery}{resolve}
\caption{Net-based Noisy GJCR for a Lipschitz \mname}
\label{alg:netngjcr}
$p_W,\delta_1,\beta$ provided, with $\delta_1\leq 1/(k(k+1))-\alpha/k^2$ for a constant $\alpha>0$ and $k\leq\beta\leq|\mathcal{N}|$\;
$\rho\gets \alpha/(4Kdk^2)$\;
Build a $\rho$-net $\mathcal{N}\subseteq C$ with cell map $\mathrm{rep}:C\to \mathcal{N}$ (\Cref{def:net})\;
$h_1\gets \lceil (8k^4/\alpha^2)\log(4k|\mathcal{N}|/p_W)\rceil$\;
$W\gets\varnothing$;\quad $\ell\gets k$\;
\While{$\ell\geq 1$}{
    $q^*\gets(2k+1)\ell/(2k(k+1))-\delta_1/2$\;
    Partition $\mathcal{N}$ into batches $\mathcal{N}_1,\dots,\mathcal{N}_{\lceil|\mathcal{N}|/\beta\rceil}$ of at most $\beta$ candidates\;
    \FQuery{$v,W\cup \mathcal{N}_t$} for each $t$ and for $h_1$ fresh voters per batch\;
    $\zeta_{c'}\gets|\{v\in\text{pool of }\mathcal{N}_{t(c')}:\ c'\in A(v),\ |A(v)\cap W|<\ell\}|$ for each $c'\in \mathcal{N}$\;
    \eIf{there is $c\in C\setminus W$ with $\zeta_{\mathrm{rep}(c)}/h_1\geq q^*$}{
        $W\gets W\cup\{c\}$ for one such $c$ (preferring $c=\mathrm{rep}(c)$ when available)\;
    }{
        $\ell\gets\ell-1$\;
    }
}
\Return{$W$}\;
\end{algorithm}
\section{Proofs from \Cref{sec:wselection}}
\label{sec:pgjcrproof}
\begin{remark}[Strict monotonicity, WLOG]
\label{rem:strictmono}
We may assume every CDF (true or estimated) fed into a {\qsl} call is strictly increasing: replacing $\widehat F^i_a$ by $\widehat F^i_a(x)+\iota x$ for an arbitrarily small $\iota>0$ preserves monotonicity, makes it strictly increasing, and changes its (PF) error by at most $\iota$, which we absorb into $\varepsilon$ since $\iota$ can be taken arbitrarily small; similarly for $\widehat F^i_b$ and $\widehat F_a$.
Consequently, ties in {\qsl}'s selection criterion (\Cref{alg:quantile-select}) arise only between distinct candidates that happen to share the exact same coordinate on the relevant axis; such ties are broken by an arbitrary fixed rule (e.g.\ candidate index), and none of our results depend on which candidate is chosen.
\end{remark}
\begingroup
\def\thetheorem{\ref{thm:pgjcr}}
\def\theHtheorem{restate-thm-pgjcr}
\begin{theorem}
\pgjcrstatement
\end{theorem}
\addtocounter{theorem}{-1}
\endgroup
\begin{proof}
First, we show that $|W|\leq k$.
Let $p_X$ be the probability that $A(v)=X$ defined over the sets of $C$.
We start by considering the function $\theta_0(X)=(k+1)p_X$, so $\sum_{X\in \mathcal{P}(C)}\theta_0(X)=k+1$.
To construct the payments, consider any step in which a candidate $c$ is selected: say that this is the $i$'th candidate to be added to $W$.
At this step, let $\mathcal{X}_c$ denote the set $\{X\in \mathcal{P}(C):c\in X,|W\cap X|<\ell\}$.
Since $c$ was added, we know that $\Pr(A(v)\in \mathcal{X}_c)>\frac{\ell}{k+1}$.
Let $p_c=\Pr(A^*(v)\in \mathcal{X}_c)=\sum_{X\in \mathcal{X}_c}p_X$.
We define two new functions, $\pi_c(X)=p_X/p_c$ and $\theta_i(X)=\theta_{i-1}(X)-\pi_c(X)$.
We have that $\sum _{X\in \mathcal{X}_c} \pi_c(X)=1$.
We notice that in any previous iteration that considered $c'$, say, we have $p_{c'}>\frac{\ell}{k+1}$.
Since for each $X\in \mathcal{X}_c$ we have $|X\cap W|<\ell$, we have that $\theta_i(X)>\theta_0(X)-\ell p_X(k+1)/\ell=0$.
Consider $\theta_k$ (or the last $\theta_i$ if fewer than $k$ candidates are selected for $W$ in total).
In each iteration, $\sum_{X\in \mathcal {P}(C)}\theta_i(X)=(\sum _{X\in \mathcal {P}(C)}\theta_{i-1}(X))-1$, $\sum _{X\in \mathcal {P}(C)}\theta_k(X)>0$, and $\sum _{X\in \mathcal {P}(C)}\theta_0(X)=k+1$, we must have at most $k$ iterations, and thus at most $k$ candidates selected.

Now, we show that $W$ satisfies (PW) with margin $1/(k(k+1))$. 
We shall show for any $c \in C\setminus W$, $\ell\in[k]$, that $p_{c,\ell}\leq \ell/(k+1)\leq \ell/k-1/(k(k+1))$. 
By construction, if $p_{c,\ell}> \ell/(k+1)$ then $c$ would have been selected for $W$, hence this cannot be the case.
\end{proof}
\begingroup
\def\thetheorem{\ref{thm:ngjcr}}
\def\theHtheorem{restate-thm-ngjcr}
\begin{theorem}
\ngjcrstatement
\end{theorem}
\addtocounter{theorem}{-1}
\endgroup
The concentration argument below adapts the analysis Brill and Peters give for their Noisy GJCR algorithm \citep{DBLP:conf/sigecom/Brill023} to the (PW) notion and the empirical statistic $\zeta_c$ used here.
\begin{proof}
Fix a level $\ell\in[k]$, call $c\in C$ \emph{tiny} if $p_{c,\ell}\leq \ell/(k+1)$ and \emph{large} if $p_{c,\ell}\geq \ell/k-\delta_1$, and write $g\coloneqq \ell/(2k(k+1))-\delta_1/2$ for the gap separating $q^*$ from each of these two bounds.
At this level, NGJCR admits a candidate only via the test $\zeta_c/h_1\geq q^*$, so it can violate (PW) with margin $\delta_1$ only by admitting some tiny $c$ or rejecting some large $c$; by the counting argument of \Cref{thm:pgjcr}, rejecting every large candidate also keeps $|W|\leq k$.
It therefore suffices to bound each of these two events, for fixed $(c,\ell)$, by $\frac{p_W}{4mk}$: a union bound over the at most $mk$ pairs $(c,\ell)$ then caps the total failure probability at $\frac{p_W}{2}\leq p_W$.

For tiny $c$, $\mathbb{E}[\zeta_c]\leq h_1\ell/(k+1)=h_1(q^*-g)$, so admission requires $\zeta_c$ to exceed its mean by at least $h_1g$; by \Cref{fact:hoeffding},
\[
\Pr[\zeta_c\geq h_1q^*]\leq\exp(-2g^2h_1).
\]
For large $c$, symmetrically $\mathbb{E}[\zeta_c]\geq h_1(\ell/k-\delta_1)=h_1(q^*+g)$, so rejection requires a deviation of at least $h_1g$ below the mean, giving the same bound $\Pr[\zeta_c<h_1q^*]\leq\exp(-2g^2h_1)$.
As noted above, $g\geq \alpha/k^2$, so both bounds are at most $\exp(-2h_1\alpha^2/k^4)$, which equals $\left(\frac{p_W}{4mk}\right)^2\leq\frac{p_W}{4mk}$ for our choice of $h_1$.
\end{proof}
\subsection{Proof of \Cref{thm:1dpry}}
First, we prove two additional properties of {\qsl}.\begin{lemma}[Pigeonhole]
\label{lem:1dpigeonhole}
Let $S={\qsl}(C,\qsarg,G)$, sorted by position as $u_{(1)},\ldots,u_{(|S|)}$.
Then for every $i\in[|S|]$,
\[
    G(u_{(i)})\geq i\qsarg.
\]
\end{lemma}
In words, since {\qsl} fills distinct quantile levels as it scans from left to right, the $i$-th selected point (by position) has already cleared $i$ of them, so its $G$-value is at least $i\qsarg$; the selected points fan out across the range rather than bunching together.
\begin{proof}
Each $u\in S$ was selected at some step $r(u)\in\{1,\ldots,\lceil 1/\qsarg\rceil-1\}$, so $G(u)\geq r(u)\qsarg$.
Different elements of $S$ are selected at different steps, so $r$ is an injection.
Among $u_{(1)},\ldots,u_{(i)}$, the values $r(u_{(1)}),\ldots,r(u_{(i)})$ are $i$ distinct positive integers, so at least one is $\geq i$, and the candidate selected at that step has $G$-value $\geq i\qsarg$.
Since $u_{(i)}$ has the largest $G$-value among $u_{(1)},\ldots,u_{(i)}$, we get $G(u_{(i)})\geq i\qsarg$.
\end{proof}

\begin{lemma}[Left-count]
\label{lem:1drho}
Let $S={\qsl}(C,\qsarg,G)$, fix $c\in C\setminus S$, and let $\xi=|\{r\in\mathbb{N}:r\qsarg<G(c)\}|$.
Then $|\{u\in S:u<c\}|\geq\xi$.
\end{lemma}
In words, every quantile level strictly below $G(c)$ is claimed by a selected point lying to the left of $c$, so the larger $G(c)$ is, the more selected points are guaranteed to sit to its left.
\begin{proof}
We show that each quantile level $r\qsarg$ with $r\leq\xi$ selects a candidate at position strictly less than $c$.
Since distinct levels select distinct candidates, this gives $|\{u\in S:u<c\}|\geq\xi$.

Fix such an $r$.
At step $r$, candidate $c$ is still available ($c\notin S$) and satisfies $G(c)>r\qsarg$, so $c\in Q_r$.
The point $u_r=\arg\min_{c'\in Q_r}G(c')$ selected at step $r$ satisfies $G(u_r)\leq G(c)$.
Since $G$ is strictly increasing (\Cref{rem:strictmono}) and $C\subset\mathbb{R}$ is a set of distinct values, $G(u_r)=G(c)$ would force $u_r=c$, contradicting $u_r\in S$ while $c\notin S$; so $G(u_r)<G(c)$ strictly, and $u_r<c$ by monotonicity of $G$.
\end{proof}

Now to prove \Cref{thm:1dpry}, we introduce some notation.
For $c\in C\setminus\W$, write $\sigma=|\{w\in\W:w<c\}|$ for the number of $\W$-candidates strictly to the left of $c$, and sort $\W=\{w_1<w_2<\cdots<w_{|\W|}\}$ by position.
Let $w^-_j$ ($j\geq1$) denote the $j$-th candidate of $\W$ to the left of $c$ (so $w^-_1=w_\sigma$, $w^-_2=w_{\sigma-1}$, \ldots); adopt the convention $w^-_j=0$ if $j>\sigma$.
We specialise the {\qsl} bounds of \Cref{sec:quantsel} to $S=\W$, $G=\widehat{F}_a$, and $\qsarg=1/(k+1)$.
The pigeonhole bound (\Cref{lem:1dpigeonhole}) gives $\widehat{F}_a(w_{(i)})\geq i/(k+1)$ for the $i$-th $\W$-candidate by position, and the left-count bound (\Cref{lem:1drho}) gives $\sigma\geq|\{r\in[k]:r/(k+1)<\widehat{F}_a(c)\}|$.

\begin{lemma}[Distance bound]
\label{lem:1ddist}
For every $c\in C\setminus\W$ and every $\ell\in[k]$,
\[
    \widehat{F}_a(c)-\widehat{F}_a(w_\ell^-)\leq\frac{\ell}{k+1}.
\]
Consequently, $F_a(c)-F_a(w_\ell^-)\leq\ell/(k+1)+2\varepsilon$.
\end{lemma}
\begin{proof}
Let $\xi=|\{r\in[k]:r/(k+1)<\widehat{F}_a(c)\}|$.
Then $\widehat{F}_a(c)\leq(\xi+1)/(k+1)$.
\Cref{lem:1drho} gives $\sigma\geq\xi$.

If $\sigma\geq\ell$: $w^-_\ell=w_{(\sigma-\ell+1)}$ in position-sorted order.
\Cref{lem:1dpigeonhole} gives $\widehat{F}_a(w_\ell^-)\geq(\sigma-\ell+1)/(k+1)$.
Hence
\[
\widehat{F}_a(c)-\widehat{F}_a(w_\ell^-)
\leq\frac{\xi+1}{k+1}-\frac{\sigma-\ell+1}{k+1}
=\frac{\xi-\sigma+\ell}{k+1}
\leq\frac{\ell}{k+1},
\]
using $\xi\leq\sigma$.

If $\sigma<\ell$: by convention $w_\ell^-=0$ and $\widehat{F}_a(w_\ell^-)=0$.
Since $\sigma\geq\xi$ and $\sigma<\ell$, we have $\xi\leq\ell-1$, so $\widehat{F}_a(c)\leq(\xi+1)/(k+1)\leq\ell/(k+1)$.

The true-CDF bound follows by (PF): $F_a(c)-F_a(w_\ell^-)\leq\widehat{F}_a(c)-\widehat{F}_a(w_\ell^-)+2\varepsilon\leq\ell/(k+1)+2\varepsilon$.
\end{proof}

\begin{lemma}[Inclusion]
\label{lem:1dinclusion}
$N_{c,\ell}\subseteq\{v:a_v\in(w_\ell^-,c]\}$.
\end{lemma}
\begin{proof}
If $c\in A(v)$ then $a_v\leq c$.
If, additionally, $|A(v)\cap\W|\leq\ell-1$, then the number of $\W$-candidates approved on the left of $c$ is at most $\ell-1$, so $a_v$ lies strictly to the right of the $\ell$-th $\W$-candidate on the left, i.e.\ $a_v>w_\ell^-$.
\end{proof}
\begingroup
\def\thetheorem{\ref{thm:1dpry}}
\def\theHtheorem{restate-thm-1dpry}
\begin{theorem}[1-D W-selection]
\onedprystatement
\end{theorem}
\addtocounter{theorem}{-1}
\endgroup
\begin{proof}
By \Cref{lem:1dinclusion}, $p_{c,\ell}\leq\Pr[a_v\in(w_\ell^-,c]]=F_a(c)-F_a(w_\ell^-)$, where the last equality is the definition of $F_a$ and holds for any joint distribution on $(a_v,b_v)$.
\Cref{lem:1ddist} gives the bound.
\end{proof}
\section{Proof of \Cref{prop:approximateapprovals}}
\label{sec:ddimquery}
To prove \Cref{prop:approximateapprovals}, we first prove that the estimated endpoints bracket the true endpoints on each dimension.
\begin{lemma}
\label{lem:boundsonapp}
$\hat{a}^i_v< a^i_v\leq \check{a}^i_v$, and $\check{b}^i_v\leq b^i_v<\hat{b}^i_v$.
\end{lemma}
\begin{proof}
We'll show that $\hat{a}^i_v< a^i_v\leq \check{a}^i_v$ and $\check{b}^i_v\leq b^i_v<\hat{b}^i_v$ will follow by symmetry.
First, $\hat{a}^i_v< a^i_v$: if $\hat{a}^i_v$ is the trailing sentinel $-\infty$ this is immediate; otherwise $\hat{a}^i_v$ is a candidate that was outputted, so \planar$(\hat{a}^i_v,v,i,>)$ is true and thus $a^i_v>\hat{a}^i_v$.
Now we will show $a^i_v\leq \check{a}^i_v$.
If $\check{a}^i_v$ is the leading sentinel $+\infty$ this is immediate; otherwise, since $\check{a}^i_v> P[x]$, it must be the case that \planar$(\check{a}^i_v,v,i,>)$ is false, so $a^i_v\not>\check{a}^i_v$, so $a^i_v\leq \check{a}^i_v$.
\end{proof}
Now we can prove \Cref{prop:approximateapprovals}.
\begingroup
\def\theproposition{\ref{prop:approximateapprovals}}
\def\theHproposition{restate-prop-approximateapprovals}
\begin{proposition}
\approxapprovalsstatement
\end{proposition}
\addtocounter{theorem}{-1}
\endgroup
\begin{proof}
$\check{A}(v)\subseteq A(v)\subseteq \hat{A}(v)$ follows from applying \Cref{lem:boundsonapp} on each $i\in[d]$.
Now we show that $A(v)\cap P=\hat{A}(v)\cap P$.
Clearly we must have $A(v)\cap P\subseteq \hat{A}(v)\cap P$ and so we need to show that $A(v)\cap P\supseteq\hat{A}(v)\cap P$.
Let $u\in \hat{A}(v)\cap P$, so $u\in P$, meaning we only need to show that $u\in A(v)$.
Suppose to the contrary that $u\not\in A(v)$ so for some $i\in[d]$ we have, WLOG again, $u_i<a^i_v$.
But $u\in \hat{A}(v)$, and so $\hat{a}^i_v<u_i<\hat{b}^i_v$ so $\hat{a}^i_v<u_i<a^i_v$ but then since $u_i<a^i_v$, $\hat{a}^i_v$ would not have been selected by \texttt{outerbound}, a contradiction.
\end{proof}

\section{Proof of \Cref{thm:multidim}}
\label{app:multidim}
We first prove the bracketing-gap bound (\Cref{lem:quantile-gaps}) used throughout the analysis.
\begingroup
\def\thelemma{\ref{lem:quantile-gaps}}
\def\theHlemma{restate-lem-quantile-gaps}
\begin{lemma}
\quantilegapsstatement
\end{lemma}
\addtocounter{theorem}{-1}
\endgroup
\begin{proof}
We prove the first bound; the second is symmetric.

Recall that $P_i\supseteq{\qsl}(C,\Delta,G)$ for $G\coloneqq 1-\widehat{F}^i_a$, the call that controls upward gaps of $\widehat{F}^i_a$.
By Coverage (\Cref{lem:lefty}), either $G(c)<\Delta$, or there is $u\in P_i$ with $G(u)\leq G(c)<G(u)+\Delta$.

In the first case, $\widehat F^i_a(c_i)>1-\Delta$, and (PF) gives $F^i_a(c_i)>\widehat F^i_a(c_i)-\varepsilon>1-\Delta-\varepsilon$; taking $c^+_i=+\infty$ (so $F^i_a(c^+_i)=1$) gives $F^i_a(c^+_i)-F^i_a(c_i)<\Delta+\varepsilon<\Delta+2\varepsilon$ directly.

In the second case, $\widehat F^i_a(u_i)\geq\widehat F^i_a(c_i)$, and since $\widehat F^i_a$ is strictly increasing (\Cref{rem:strictmono}), $u_i\geq c_i$.
As $c^+_i$ is the smallest $P_i$-coordinate $\geq c_i$ and $u\in P_i$, $c^+_i\leq u_i$, hence $F^i_a(c^+_i)\leq F^i_a(u_i)$.
Passing between $\widehat{F}^i_a$ and $F^i_a$ with (PF) ($|\widehat{F}^i_a-F^i_a|\leq\varepsilon$), we get $
    {F^i_a(c^+_i)\leq F^i_a(u_i)
    \leq\widehat{F}^i_a(u_i)+\varepsilon
    <\widehat{F}^i_a(c_i)+\Delta+\varepsilon
    \leq F^i_a(c_i)+\Delta+2\varepsilon.}
$
The second bound is symmetric, replacing the call on $1-\widehat{F}^i_a$ by the call on $\widehat{F}^i_b$ and the right bracket $c^+_i\geq c_i$ by the left bracket $c^-_i\leq c_i$.
\end{proof}
\begingroup
\def\thetheorem{\ref{thm:multidim}}
\def\theHtheorem{restate-thm-multidim}
\begin{theorem}[Query efficiency]
\multidimstatement
\end{theorem}
\addtocounter{theorem}{-1}
\endgroup
\begin{proof}
We bound the expected per-voter query count.
The argument has two parts: under (PW) and (PF) the verification test passes with high probability, so a typical voter is asked only the $\Oh(d\log dk)$ queries of one \texttt{resolve} pass over $P$, while the rare fallback contributes a negligible amount in expectation.
We first bound the failure probability, then the query cost of each branch.

Let $\Xi$ be the event that voter $v$ is counted in
$s_{\ell,c}$, i.e.\ $c\in\hat{A}(v)$ and $|\hat{A}(v)\cap\W|<\ell$; the latter
equals $|A(v)\cap\W|<\ell$ since $\W\subseteq P$.
Let $Z$ be the sub-event of $\Xi$
in which $c\notin A(v)$.
Splitting on whether $c\in A(v)$ and using $A(v)\subseteq\hat{A}(v)$ from \Cref{prop:approximateapprovals},
$
    \Pr(\Xi) \leq p_{c,\ell}+\Pr(Z)
    \leq \frac{\ell}{k}-\delta+\Pr(Z),
$
where the last inequality applies (PW) with margin $\delta$ (\Cref{def:PW-1d}).

If $c\in P$ then $\hat{A}(v)\cap P=A(v)\cap P$, so $\Pr(Z)=0$.
For $c\notin P$,
$c\in\hat{A}(v)\setminus A(v)$ requires, for some $i$, either $a^i_v>c_i$ with $\hat{a}^i_v<c_i$
(so $a^i_v\in(c_i,c^+_i)$), or $b^i_v<c_i$ with $\hat{b}^i_v>c_i$ (so $b^i_v\in(c^-_i,c_i)$),
where $c^-_i<c_i<c^+_i$ are the consecutive $P_i$-candidates bracketing $c_i$.
By \Cref{lem:quantile-gaps} (using (PF) with error $\varepsilon$),
\[
    \Pr(a^i_v\in(c_i,c^+_i))=F^i_a(c^+_i)-F^i_a(c_i)<\Delta+2\varepsilon,\qquad
    \Pr(b^i_v\in(c^-_i,c_i))=F^i_b(c_i)-F^i_b(c^-_i)<\Delta+2\varepsilon.
\]
Union bounding over $i\in[d]$, we get $\Pr(Z)\leq 2d(\Delta+2\varepsilon)$.
With the derived spacing $\Delta=\delta/(4d)-\varepsilon$ we have $\gamma\coloneqq\delta-2d(\Delta+2\varepsilon)=\frac{1}{2}(\delta-4d\varepsilon)\geq\alpha/k^2>0$ by the hypothesis.
Since $\Pr(Z)\leq \delta-\gamma$, we have $\Pr(\Xi)\leq\ell/k-\delta+\Pr(Z)\leq\ell/k-\gamma$.
\Cref{fact:hoeffding} gives
$
    \Pr(s_{\ell,c}\geq n\ell/k)\leq\exp(-2n\gamma^2)\leq\exp(-2n\alpha^2/k^4).
$
By \Cref{lem:cells}, $s_{\ell,c}=s_{\ell,c'}$ for all $P$-equivalent $c'$, so a union bound over the $\Oh(\Lambda k)$ cell--level pairs yields $\Pr(\text{verification fails})\leq\Lambda k\exp(-2n\alpha^2/k^4)$.

Since $\Delta=\delta/(4d)-\varepsilon\geq\alpha/(2dk^2)$ by the hypothesis, each $P_i$ contains at most $2/\Delta\leq 4dk^2/\alpha=O(dk^2)$ candidates from the two {\qsl} calls,
so $|P|\leq k+\Oh(d^2k^2)=\Oh(d^2k^2)$.
By \Cref{prop:queryingqueries},
querying uses $\Oh(d\log|P|)=\Oh(d\log dk)$ queries per voter, with an extra
$\Oh(d\log m)$ if verification fails.
When $n=\Omega(\alpha^{-2}k^4\log (\Lambda k\log m))=\Omega(\alpha^{-2}k^4\log(\Lambda\log m))$, we have $\Lambda k\exp(-2n\alpha^2/k^4)\leq\Lambda k\exp(-\log(\Lambda k\log m))=1/\log m$, so verification fails with probability at most $\Oh(1/\log m)$, and the expected per-voter query count is $\Oh(d\log dk)+\Lambda k\exp(-2n\alpha^2/k^4)\cdot\Oh(d\log m)=\Oh(d\log dk)+\Oh(d)=\Oh(d\log dk)$.
\end{proof}

\section{The batch estimator: full proof of \Cref{lem:logk-block}}
\label{app:logk}

\begingroup
\def\thelemma{\ref{lem:logk-block}}
\def\theHlemma{restate-lem-logk-block}
\begin{lemma}
\logkblockstatement
\end{lemma}
\addtocounter{theorem}{-1}
\endgroup
\begin{proof}
    Fix a dimension $i$, and consider estimating the CDF of the left endpoint $a_i$; the right endpoint $b_i$ is analogous. Let $F$ be the true CDF of $a_i$.
    Each candidate is assigned to a unique batch, that batch $B$ is served by a fresh pool $V_B$ of $h_2$ i.i.d.\ voters resolved only on $B$. We then estimate a CDF of $B$ based on $V_B$: for $c\in B$, $\hat{F}_B(c)=\frac{1}{h_2}\left|\{v\in V_B:a^i_v\leq c_i\}\right|$; since every $v\in V_B$ was resolved over $B$, \Cref{lem:boundsonapp} gives endpoints $\hat a^i_v,\check a^i_v$ that are consecutive $B$-candidates (or sentinels) with $\hat a^i_v<a^i_v\leq\check a^i_v$, so no $B$-candidate lies strictly between them; hence every $c\in B$ has either $c_i\leq\hat a^i_v$ (so $a^i_v>c_i$) or $c_i\geq\check a^i_v$ (so $a^i_v\leq c_i$), and we know precisely whether $a^i_v\leq c_i$.

    Consider now $F^*_B$, a (hypothetical) CDF estimate constructed by observing $a^i_v$ directly for every $v\in V_B$. Then $F^*_B(c)=\hat{F}_B(c)$ for every $c\in B$. The DKW inequality (\Cref{fact:dkw}) states \[
    \Pr\left[\max_{c\in B}\left|F^*_{B}(c)-F(c)\right|\geq\varepsilon\right]
    \leq\Pr\left[\sup_{x\in\mathbb{R}}\left|F^*_{B}(c)-F(c)\right|\geq\varepsilon\right]
    \leq 2\exp(-2h_2\varepsilon^2),
\]
and since $F^*_B(c)=\hat{F}_B(c)$ for every $c\in B$, we have $\Pr\left[\max_{c\in B}\left|\hat{F}_{B}(c)-F(c)\right|\geq\varepsilon\right]\leq 2\exp(-2h_2\varepsilon^2)$.

Define the \emph{good event}
\[
    \mathcal{E}:=\bigcap_{B}\left\{\max_{c\in B}\left|\widehat{F}_{B}(c)-F(c)\right|\leq\varepsilon\right\}.
\]
By Boole's inequality over the $m/\beta$ batches (independence of the pools is not needed),
\[
    \Pr[\mathcal{E}^c]\leq\frac{m}{\beta}\cdot 2\exp(-2h_2\varepsilon^2)=\frac{2m}{\beta}\exp(-2h_2\varepsilon^2).
\]

We now show that when we construct the estimate $\hat{F}(c)=\max_{c'\in C,c'_i\leq c_i}\hat{F}_{B(c')}(c')$ over all blocks, where $B(c)$ is the block containing $c$, the estimation property is preserved.

Condition on $\mathcal{E}$, so $|\widehat{F}_{B(c')}(c')-F(c')|\leq\varepsilon$ for every candidate $c'$.
Fix a candidate $c$ and let $c^*\in C$, $c^*_i\leq c_i$, attain the maximum defining $\widehat{F}(c)$, so $\widehat{F}(c)=\widehat{F}_{B(c^*)}(c^*)$.

\emph{Lower bound.} The maximand includes $c'=c$, so $\widehat{F}(c)\geq\widehat{F}_{B(c)}(c)\geq F(c)-\varepsilon$.

\emph{Upper bound.} At $c^*$, whose batch $B(c^*)$ is also good (this is exactly why $\mathcal{E}$ intersects over \emph{all} batches, since $c^*$ may lie in a different batch from $c$) we have $\widehat{F}_{B(c^*)}(c^*)\leq F(c^*)+\varepsilon$; and $c^*_i\leq c_i$ with $F$ non-decreasing gives $F(c^*)\leq F(c)$.
Hence $\widehat{F}(c)\leq F(c)+\varepsilon$.

Combining the two bounds, $|\widehat{F}(c)-F(c)|\leq\varepsilon$ for every $c\in C$.
Moreover $\widehat{F}$ is non-decreasing in $c_i$ and $[0,1]$-valued by construction, hence a valid CDF estimate.

So, on $\mathcal{E}$ we have $\max_{c\in C}|\widehat{F}(c)-F(c)|\leq\varepsilon$, so
\[
    \Pr\left[\max_{c\in C}\left|\widehat{F}(c)-F(c)\right|>\varepsilon\right]\leq\Pr[\mathcal{E}^c]\leq\frac{2m}{\beta}\exp(-2h_2\varepsilon^2),
\]
and the identical argument bounds $\widehat{F}^i_b$.
This proves \Cref{lem:logk-block}; \Cref{thm:logk-estimator} then follows by the union bound over the $d$ dimensions and both endpoints and the query-cost accounting given after the lemma in \Cref{sec:festimation}.
\end{proof}

%%%%%%%%%%%%%%%%%%%%%%%%%%%%%%%%%%%%%%%%%%%%%%%%%%%%%
\section{Proofs from \Cref{sec:lipschitz}}
\label{app:lipschitz}
\begingroup
\def\thelemma{\ref{lem:lipschitz-cdf}}
\def\theHlemma{restate-lem-lipschitz-cdf}
\begin{lemma}
\lipcdfstatement
\end{lemma}
\addtocounter{theorem}{-1}
\endgroup
\begin{proof}
    Fix $i\in[d]$; we treat $a^i_v$, the argument for $b^i_v$ being identical.
As $\mathcal{D}$ has density $f\leq K$, every event has probability at most $K$ times its Lebesgue measure.
For $0\leq y\leq x\leq 1$, the event $\{y< a^i_v\leq x\}$ is the slab $\{(a_1,b_1,\dots,a_d,b_d)\in\mathcal{R}: y< a_i\leq x\}$, whose measure is at most $x-y$ (one coordinate confined to an interval of length $x-y$, the rest to $[0,1]$).
Hence
    $
        F^i_a(x)-F^i_a(y)=\Pr[y< a^i_v\leq x]\leq K(x-y),
    $
    and as $F^i_a$ is non-decreasing, $|F^i_a(x)-F^i_a(y)|\leq K|x-y|$.
\end{proof}

\begingroup
\def\thelemma{\ref{lem:lip-stability}}
\def\theHlemma{restate-lem-lip-stability}
\begin{lemma}[Lipschitz stability of cohesive mass]
\lipstabilitystatement
\end{lemma}
\addtocounter{theorem}{-1}
\endgroup
\begin{proof}
Fix $W,\ell$ and let $J=\{|A(v)\cap W|<\ell\}$, an event that does not depend on the point being tested.
Then $p_{c,\ell}=\Pr[\{c\in A(v)\}\cap J]$, so
\[
    p_{c,\ell}-p_{c',\ell}
    =\Pr[\{c\in A(v)\}\cap J]-\Pr[\{c'\in A(v)\}\cap J]
    \leq\Pr\left[\{c\in A(v)\}\setminus\{c'\in A(v)\}\right].
\]
Recall $c\in A(v)\iff a^i_v\leq c_i\leq b^i_v$ for all $i\in[d]$.
If $c\in A(v)$ but $c'\notin A(v)$, some coordinate $i$ witnesses the failure, so $\{c\in A(v)\}\setminus\{c'\in A(v)\}\subseteq\bigcup_{i\in[d]}\mathcal{E}_i$ with $\mathcal{E}_i=\{a^i_v\leq c_i\leq b^i_v\}\cap\{c'_i\notin[a^i_v,b^i_v]\}$.
We bound $\Pr[\mathcal{E}_i]$ by cases on coordinate $i$.
If $c'_i<c_i$: from $c'_i<c_i\leq b^i_v$ the only way to have $c'_i\notin[a^i_v,b^i_v]$ is $c'_i<a^i_v$, i.e.\ $a^i_v\in(c'_i,c_i]$, so $\Pr[\mathcal{E}_i]\leq F^i_a(c_i)-F^i_a(c'_i)\leq K|c_i-c'_i|$ by \Cref{lem:lipschitz-cdf}.
If $c'_i>c_i$: symmetrically $b^i_v\in[c_i,c'_i)$ and $\Pr[\mathcal{E}_i]\leq K|c_i-c'_i|$.
If $c'_i=c_i$ then $\mathcal{E}_i=\varnothing$.
A union bound gives $p_{c,\ell}-p_{c',\ell}\leq K\sum_{i\in[d]}|c_i-c'_i|=K\|c-c'\|_1$, and the statement follows by symmetry in $c,c'$.
\end{proof}
\begingroup
\def\thetheorem{\ref{thm:netngjcr}}
\def\theHtheorem{restate-thm-netngjcr}
\begin{theorem}[Net-NGJCR]
\netngjcrstatement
\end{theorem}
\addtocounter{theorem}{-1}
\endgroup
\begin{proof}
Write $\tau=\alpha/(4k^2)$.
By the choice $\rho=\alpha/(4Kdk^2)$ we have $Kd\rho=\alpha/(4k^2)=\tau$, so \Cref{cor:net-transfer} reads $|p_{c,\ell}-p_{\mathrm{rep}(c),\ell}|\leq\tau$ for all $W,\ell,c$.

\emph{Geometry of the threshold.} For any $\ell\in[k]$, $q^*$ is the midpoint of $\ell/(k+1)$ and $\ell/k-\delta_1$, so
\[
    \Gamma\coloneqq q^*-\frac{\ell}{k+1}=\left(\frac{\ell}{k}-\delta_1\right)-q^*=\frac{1}{2}\left(\frac{\ell}{k(k+1)}-\delta_1\right)\geq\frac{1}{2}\left(\frac{1}{k(k+1)}-\delta_1\right)\geq\frac{\alpha}{2k^2},
\]
using $\ell\geq1$ and $\delta_1\leq 1/(k(k+1))-\alpha/k^2$.
Hence $\Gamma\geq 2\tau$.

\emph{Termination and the good event.} Each round either adds a candidate from $C\setminus W$ (a distinct candidate each time, from the finite set $C$) or decrements $\ell$, so the loop halts after finitely many rounds and returns a committee whatever the voters answer; we have not yet bounded the round count.
Fix a round with committee $W$ and level $\ell$ and a net point $c'\in\mathcal{N}$; its count $\zeta_{c'}$ is read from the \emph{fresh} pool of the batch $\mathcal{N}_{t(c')}$ containing it, a sum of $h_1$ i.i.d.\ Bernoulli$(p_{c',\ell})$ variables (each pooled voter is resolved on $W\cup \mathcal{N}_{t(c')}\ni c'$, so $\mathds{1}[c'\in A(v)]$ and $|A(v)\cap W|$ are determined).
Because this pool is fresh (independent of all earlier rounds), the bound holds whatever committee $W$ and level $\ell$ the round carries, so \Cref{fact:hoeffding} gives $\Pr[\,|\zeta_{c'}/h_1-p_{c',\ell}|\geq\tau\,]\leq 2e^{-2h_1\tau^2}$.
Let $\mathcal{E}$ be the event that every net estimate in each of the first $2k$ rounds that occur is within $\tau$.
Summing this bound over the at most $2k$ such rounds and $|\mathcal{N}|$ net points, and using $h_1\geq\log(4k|\mathcal{N}|/p_W)/(2\tau^2)$, $\Pr[\mathcal{E}^{c}]\leq 4k|\mathcal{N}|e^{-2h_1\tau^2}\leq p_W$.
This bound uses only freshness and \Cref{fact:hoeffding}; it presupposes no bound on the number of rounds or on $|W|$.
Condition on $\mathcal{E}$ for the remainder; we show next that on $\mathcal{E}$ the loop stops within $2k$ rounds, so it only ever reads the accurate pools just counted.

\emph{Size $\leq k$ and the round count.} A candidate $c$ is added only in a round where $\zeta_{\mathrm{rep}(c)}/h_1\geq q^*$; write $W_c$ for the committee at that moment.
On $\mathcal{E}$ such a round, being among the first $2k$, is accurate, so $p^{W_c}_{\mathrm{rep}(c),\ell}>q^*-\tau$, and by \Cref{cor:net-transfer}, $p^{W_c}_{c,\ell}>q^*-\tau-Kd\rho=q^*-2\tau\geq q^*-\Gamma=\ell/(k+1)$.
Thus every candidate added in the first $2k$ rounds has $p^{W_c}_{c,\ell}>\ell/(k+1)$.
The budget argument of \Cref{thm:pgjcr} bounds the number of selections with this own-mass-at-selection property by $k$..
Hence at most $k$ of the first $2k$ rounds are adds; as at most $k$ rounds decrement $\ell$ (it falls from $k$ to $0$), the loop must exit at or before round $2k$.
So the whole execution reads only the accurate pools of these $\leq 2k$ rounds, and the returned committee has size $\leq k$.

\emph{(PW) with margin $\delta_1$.} Let $W^\ast$ be the returned committee.
Fix $\ell\in[k]$ and $c\in C\setminus W^\ast$.
The loop left level $\ell$ at the round that decremented $\ell$; let $W_\ell$ be the committee then, and $p^{W_\ell}_{c,\ell}=\Pr[c\in A(v),|A(v)\cap W_\ell|<\ell]$.
That decrement happened precisely because no candidate in $C\setminus W_\ell$ was selectable, so in particular $\zeta_{\mathrm{rep}(c)}/h_1<q^*$ (note $c\in C\setminus W^\ast\subseteq C\setminus W_\ell$ since $W_\ell\subseteq W^\ast$).
On $\mathcal{E}$, $p^{W_\ell}_{\mathrm{rep}(c),\ell}<q^*+\tau$, and by \Cref{cor:net-transfer},
\[
    p^{W_\ell}_{c,\ell}<q^*+\tau+Kd\rho=q^*+2\tau\leq q^*+\Gamma=\frac{\ell}{k}-\delta_1 .
\]
Finally $p_{c,\ell}$ is non-increasing under enlarging $W$ (adding a candidate can only enlarge $|A(v)\cap W|$, shrinking the event $\{|A(v)\cap W|<\ell\}$), and $W_\ell\subseteq W^\ast$, so $p^{W^\ast}_{c,\ell}\leq p^{W_\ell}_{c,\ell}<\ell/k-\delta_1$.
As $c,\ell$ were arbitrary, $W^\ast$ satisfies (PW) with margin $\delta_1$.

\emph{Queries.} Each pooled voter is resolved once, on $W\cup \mathcal{N}_{t}$, a set of size at most $k+\beta$; by \Cref{prop:queryingqueries} this is $\Oh(d\log(k+\beta))=\Oh(d\log\beta)$ \planar\ queries per voter (using $\beta\geq k$).
Each round serves $\lceil|\mathcal{N}|/\beta\rceil$ batches of $h_1$ fresh voters, and there are at most $2k$ rounds, so the pool used is at most $2k\,h_1\lceil|\mathcal{N}|/\beta\rceil$ and the total is $\Oh\left(k\,h_1\,(|\mathcal{N}|/\beta)\,d\log\beta\right)$ queries.
With a single batch $\beta=|\mathcal{N}|$ and $\alpha$ constant, $q_W=\Oh(d\log|\mathcal{N}|)=\Oh(d^2\log(Kdk))$; and substituting $|\mathcal{N}|\leq\lceil 4Kdk^2/\alpha\rceil^d$ into $h_1$, so that $\log(4k|\mathcal{N}|/p_W)=\Oh(d\log(Kdk)+\log(1/p_W))$, gives $n_W=\Theta(k\,h_1)=\Oh\left(k^5\left(d\log(Kdk)+\log(1/p_W)\right)\right)$.
\end{proof}

\begingroup
\def\thetheorem{\ref{thm:lipschitz-estimator}}
\def\theHtheorem{restate-thm-lipschitz-estimator}
\begin{theorem}
\lipestimatorstatement
\end{theorem}
\addtocounter{theorem}{-1}
\endgroup
\begin{proof}
    Since ${\qsl}$ places grid points at the quantile levels $\varepsilon/2K,\,2\varepsilon/2K,\dots$ of the uniform distribution, we have $|T|\leq 2K/\varepsilon$ and, by coverage (\Cref{lem:lefty}), every $c\in C$ satisfies $|c_i-(\lefty{c}{T})_i|\leq\varepsilon/2K$ on each axis $i$.

    Partition $T$ into batches of at most $\beta$ grid points; serve batch $B$ with its own fresh pool $V_B$ of $h_2$ voters, resolved only on $B$, so each pooled voter answers $\Oh(d\log\beta)$ queries (\Cref{prop:queryingqueries}).
Exactly as in the $\Oh(d\log k)$ estimator (\Cref{thm:logk-estimator} and \Cref{app:logk}), resolving $v\in V_B$ on $B$ determines $\mathds{1}[a^i_v\leq c_i]$ for every grid point $c\in B$ (including voters whose endpoint lies outside $B$ on axis $i$), so the per-batch value $\widehat{F}^i_{a,B}(c)=\frac{1}{h_2}|\{v\in V_B:a^i_v\leq c_i\}|$ is a genuine empirical CDF of $h_2$ i.i.d.\ samples.
We report the running maximum over the grid,
    \[
        \widehat{F}^i_a(c):=\max\left\{\widehat{F}^i_{a,B(c')}(c'):c'\in T,\ c'_i\leq c_i\right\},
    \]
    which is non-decreasing and, for $c\notin T$, coincides with the left-extension $\widehat{F}^i_a(\lefty{c}{T})$; likewise for $\widehat{F}^i_b$.

    By the Dvoretzky--Kiefer--Wolfowitz inequality (\Cref{fact:dkw}) each per-batch empirical CDF deviates from $F^i_a$ by $\varepsilon/2$ with probability at most $2\exp(-2h_2(\varepsilon/2)^2)=2\exp(-h_2\varepsilon^2/2)$.
Union bounding over the $\lceil|T|/\beta\rceil$ batches, the $d$ dimensions and both endpoints, the probability that any batch deviates by $\varepsilon/2$ is at most $4d\lceil|T|/\beta\rceil\exp(-h_2\varepsilon^2/2)\leq p_F$, using $h_2\geq 2\log(4d\lceil|T|/\beta\rceil/p_F)/\varepsilon^2$.
Condition on the complementary good event $\mathcal{E}$.
As in \Cref{lem:logk-block}, on $\mathcal{E}$ the running maximum is within $\varepsilon/2$ of $F^i_a$ at every grid point of $T$: for the grid point $c^*$ attaining the maximum (with $c^*_i\leq c_i$, possibly in another batch, so $\mathcal{E}$ must range over \emph{all} batches), monotonicity gives $\widehat{F}^i_{a,B(c^*)}(c^*)\leq F^i_a(c^*)+\varepsilon/2\leq F^i_a(c)+\varepsilon/2$ for $c\in T$, while $\widehat{F}^i_a(c)\geq\widehat{F}^i_{a,B(c)}(c)\geq F^i_a(c)-\varepsilon/2$.

    Finally, on $\mathcal{E}$, for any $c\in C$ the left-extension gives
    \begin{multline*}
        |\widehat{F}^i_a(c)-F^i_a(c)|=|\widehat{F}^i_a(\lefty{c}{T})-F^i_a(c)|\leq|\widehat{F}^i_a(\lefty{c}{T})-F^i_a(\lefty{c}{T})|+|F^i_a(\lefty{c}{T})-F^i_a(c)|\\
        \leq\varepsilon/2+K\,|(\lefty{c}{T})_i-c_i|\leq \varepsilon/2+K\cdot\varepsilon/2K=\varepsilon,
    \end{multline*}
    using $\lefty{c}{T}\in T$, the grid accuracy above, and \Cref{lem:lipschitz-cdf}; the same holds for $\widehat{F}^i_b$.
Hence (PF) holds with error $\varepsilon$.
The pool is $\lceil|T|/\beta\rceil$ disjoint sets of $h_2$ voters; $|T|\leq 2K/\varepsilon$ makes the per-voter $\Oh(d\log\beta)$ load independent of $m$ and the pool free of any polynomial-in-$m$ factor, the latter depending on $m$ only through the $\log(1/p_F)$ in $h_2$ (a $\log\log m$ once $p_F=\Theta(1/\log m)$ is fixed in \Cref{sec:synthesis}).
\end{proof}

%%%%%%%%%%%%%%%%%%%%%%%%%%%%%%%%%%%%%%%%%%%%%%%%%%%%%
\section{Proofs from \Cref{sec:synthesis}}
\label{app:synthesis}
We restate and prove the four end-to-end corollaries of \Cref{sec:synthesis}.
Each instantiates \Cref{fact:amortized} with one W-selection and one $\widehat{F}$-estimation module, whose pool sizes and per-voter query loads are read off from \Cref{sec:wselection,sec:festimation,sec:lipschitz}.

\begingroup
\def\thecorollary{\ref{cor:knowndist}}
\def\theHcorollary{restate-cor-knowndist}
\begin{corollary}[Known distribution]
\knowndiststatement
\end{corollary}
\addtocounter{theorem}{-1}
\endgroup
\begin{proof}
PGJCR supplies (PW) with margin $\delta=1/(k(k+1))$ (\Cref{thm:pgjcr}), and the exact per-dimension CDFs supply (PF) with error $\varepsilon=0$; together they meet the framework's hypothesis $\delta-4d\varepsilon\geq 2\alpha/k^2$ (\Cref{thm:multidim}) for any $\alpha\leq 1/4$.
Neither module queries voters ($n_W=n_F=0$), so \Cref{fact:amortized} reduces to the framework floor alone, $n=\Omega(k^4\log(\Lambda\log m))$, and every voter pays only the $\Oh(d\log dk)$ verification pass --- hence the per-voter bound holds for every voter w.h.p., not merely in expectation.
By \Cref{lem:cells} the cell technique gives $\log\Lambda=\Oh(d\log dk)$ (improving on the naive count of $m$ candidates once $d\log dk\Oh(\log m)$, i.e.\ $\Lambda<m$), so $\log(\Lambda\log m)=\Oh(d\log dk+\log\log m)$ and $n=\Omega\left(k^4(d\log dk+\log\log m)\right)$, with $m$-dependence only $\log\log m$.
\end{proof}

\begingroup
\def\thecorollary{\ref{cor:logm}}
\def\theHcorollary{restate-cor-logm}
\begin{corollary}[Unknown distribution, $\Oh(\log m)$ budget]
\logmstatement
\end{corollary}
\addtocounter{theorem}{-1}
\endgroup
\begin{proof}
With $\beta=\Theta(m)$, a single batch covering all of $C$, NGJCR draws a fresh pool of $n_W=\Theta(h_1 k)=\Theta(k^5\log m)$ voters, each asked $q_W=\Oh(d\log m)$ \planar\ queries, and supplies (PW) at a margin meeting the framework's requirement (\Cref{thm:ngjcr}); the batch estimator, also on a single batch, draws $n_F=h_2=\Theta(d^2k^4\log(d\log m))$ voters at $q_F=\Oh(d\log m)$ each (\Cref{thm:logk-estimator}).
Hence $$n_W q_W+n_F q_F=\Oh\left(d\,k^4\log m\,(k\log m+d^2\log(d\log m))\right);$$ dividing by $d\log dk$ and adding the framework floor $\Omega(k^4\log(\Lambda\log m))$ (\Cref{thm:multidim}) gives the exact threshold
\[
    n=\Omega\left(k^4\log(\Lambda\log m)+k^5\log^2 m+d^2k^4\log m\log\log m\right),
\]
and \Cref{fact:amortized} yields the amortized bound.
Absorbing the floor via $k^4\log(\Lambda\log m)\leq 2k^4\log m\leq k^5\log^2 m$ (as $\Lambda\leq m$) gives the simplified sufficient form $n=\Omega\left(k^5\log^2 m+d^2k^4\log m\log\log m\right)$.
\end{proof}

\begingroup
\def\thecorollary{\ref{cor:logk}}
\def\theHcorollary{restate-cor-logk}
\begin{corollary}[Unknown distribution, $\Oh(\log k)$ budget]
\logkstatement
\end{corollary}
\addtocounter{theorem}{-1}
\endgroup
\begin{proof}
With the candidates split into $\Theta(m/\beta)$ batches of size $\beta=\Theta(k^5)$, each served by a fresh pool, every pooled voter is asked only $q_W=q_F=\Oh(d\log\beta)=\Oh(d\log k)$ \planar\ queries --- the promised $\Oh(\log k)$ budget.
NGJCR's pool is then $n_W=\Theta(h_1 km/\beta)=\Theta(m\log m)$ (\Cref{thm:ngjcr}) and the estimator's is $n_F=\Theta((m/\beta)h_2)=\Oh\big((d^2m/k)\log(dm)\big)$ (\Cref{thm:logk-estimator}), where $h_2=\Theta(d^2k^4\log(4dm/(\beta p_F)))$ and $\log(4dm/(\beta p_F))=\Oh(\log(dm))$ at $p_F=\Theta(1/\log m)$.
Hence $n_W q_W=\Oh(dm\log m\log k)$ and $n_F q_F=\Oh\big((d^3m/k)\log(dm)\log k\big)$; dividing by $d\log dk$ and adding the framework floor $\Omega(k^4\log(\Lambda\log m))$ (\Cref{thm:multidim}) gives the exact threshold
\[
    n=\Omega\left(k^4\log(\Lambda\log m)+\Big(m\log m+\frac{d^2m}{k}\log(dm)\Big)\right),
\]
and \Cref{fact:amortized} yields the amortized bound.
Since $\log(\Lambda\log m)\leq 2\log m$ (as $\Lambda\leq m$), this is implied by the simplified sufficient form $n=\Omega\left((k^4+m)\log m+(d^2m/k)\log(dm)\right)$ reported in the main body.
\end{proof}

\begingroup
\def\thecorollary{\ref{cor:lipschitz}}
\def\theHcorollary{restate-cor-lipschitz}
\begin{corollary}[Lipschitz distribution, single pool]
\lipschitzcorstatement
\end{corollary}
\addtocounter{theorem}{-1}
\endgroup
\begin{proof}
Net-NGJCR on a single pool ($\beta=|\mathcal{N}|$) uses $n_W=\Theta(k\,h_1)$ voters at $q_W=\Oh(d\log|\mathcal{N}|)$ each (\Cref{thm:netngjcr}), and the grid estimator on a single pool ($\beta=|T|$) uses $n_F=h_2$ voters at $q_F=\Oh(d\log|T|)$ each (\Cref{thm:lipschitz-estimator}).
The net $\mathcal{N}$ has $\Oh((Kdk^2)^d)$ cells and the grid $T$ has $\Oh(K/\varepsilon)=\Oh(Kdk^2)$ points, both independent of $m$, with $\log|\mathcal{N}|=\Oh(d\log(Kdk))$ and $\log|T|=\Oh(\log(Kdk))$.
With the standing choices $p_W=p_F=\Theta(1/\log m)$ (\Cref{fact:amortized}) and $\varepsilon=\Theta(1/dk^2)$, the pool sizes depend on $m$ only through a $\log\log m$ (from $\log(1/p_W),\log(1/p_F)$):
\[
    h_1=\Oh\big(k^4(\log(k|\mathcal{N}|)+\log\frac{1}{p_W})\big)=\Oh\big(k^4(d\log(Kdk)+\log\log m)\big),
    h_2=\Oh\big(\frac{\log(d/p_F)}{\varepsilon^2}\big)=\Oh\big(d^2k^4\log(d\log m)\big).
\]
Hence, with $q_W=\Oh(d\log|\mathcal{N}|)=\Oh(d^2\log(Kdk))$ and $q_F=\Oh(d\log|T|)=\Oh(d\log(Kdk))$,
\[
    n_W q_W=\Oh\big(d^2k^5\log(Kdk)\,(d\log(Kdk)+\log\log m)\big),\qquad
    n_F q_F=\Oh\big(d^3k^4\log(Kdk)\,\log(d\log m)\big).
\]
Dividing by $d\log dk$ and adding the framework floor $n=\Omega(k^4\log(\Lambda\log m))$ (\Cref{thm:multidim}, so that $\eta=\Oh(1/\log m)$) gives the exact threshold
\[
    n=\Omega\left(k^4\log(\Lambda\log m)+\frac{d^2k^5\log^2(Kdk)+d\,k^5\log(Kdk)\log\log m+d^2k^4\log(Kdk)\log(d\log m)}{\log dk}\right),
\]
and \Cref{fact:amortized} yields the amortized bound.
Each pool term carries a factor $\log(Kdk)/\log dk=1+\log K/\log dk\leq 1+\log K$ (using $\log dk\geq 1$); factoring it out and bounding $\log\log m\leq\log(d\log m)$, $dk^5\leq d^2k^5$, and $k^4\leq k^5$ collapses the three terms into $d^2k^5(1+\log K)(\log(Kdk)+\log(d\log m))$, which (as shown next) also dominates the framework floor; this gives the simplified sufficient form $n=\Omega\big(d^2k^5(1+\log K)(\log(Kdk)+\log(d\log m))\big)$.
By \Cref{lem:cells} the cell technique bounds $\log\Lambda=\Oh(d\log dk)$, so the verification floor $k^4\log(\Lambda\log m)=\Oh(k^4(d\log dk+\log\log m))$ is dominated by the net pool and never binding. This is what keeps the electorate free of any polynomial-in-$m$ factor: without the cell bound the floor would be $\Theta(k^4\log m)$, whereas here the whole threshold's only $m$-dependence is the $\log\log m$ inside $\log(d\log m)$.
The verification pass costs $\Oh(d\log dk)$ regardless of $K$, so the amortized per-voter bound is $K$-free; the heavier net load $q_W=\Oh(d^2\log(Kdk))$ is paid only by the fixed pool.
\end{proof}

\end{document}